\documentclass[12pt]{article}

\usepackage[dvipsnames]{xcolor}
\usepackage{amsmath,amsthm,amssymb,natbib,verbatim,cancel, bbm, tikz, tkz-euclide, graphicx,subcaption, float,titling}

\usepackage[scr=boondoxo  
            ]   
           {mathalpha}

\usepackage{caption}
\usetikzlibrary{bending,tikzmark}
\usetikzlibrary{shapes.misc, positioning, arrows,decorations.markings}
 \usetikzlibrary{calc}

 \usepackage{tikz-3dplot}
\usepackage{pgfplots}
\usetikzlibrary{shapes}
\tdplotsetmaincoords{60}{110}

\usepgfplotslibrary{fillbetween}
\usetikzlibrary{intersections}

 \usepackage{hyperref}
\hypersetup{
    colorlinks =true,
    allcolors =NavyBlue,
    allbordercolors =white,
}

\usepackage[margin=1in]{geometry} 
\usepackage{setspace}

\tdplotsetmaincoords{70}{120} 

\theoremstyle{definition}
\newtheorem{theorem}{Theorem}
\newtheorem*{theorem*}{Theorem}
\newtheorem*{lemma*}{Lemma}

\newtheorem{proposition}{Proposition}
\newtheorem{corollary}{Corollary}

\newtheorem{lemma}{Lemma}

\newtheorem{example}{Example}

\definecolor{amber}{rgb}{1.0,0.75,0.0}
\definecolor{aqua}{rgb}{0,1,1}
\definecolor{amaranth}{rgb}{1,.6,.62}

\newcommand\aug{\fboxsep=-\fboxrule\!\!\!\fbox{\strut}\!\!\!}

\title{Identification in Stochastic Choice\thanks{We are grateful to Roy Allen, Jose Apesteguia, Pierpaolo Battigalli, Andr\'{e}s Carvajal, Simone Cerreia-Vioglio, Christopher Chambers, Emel Filiz-Ozbay, Satoshi Fukuda, Christopher Kops, SangMok Lee, Jay Lu, Yusufcan Masatlioglu, Paulo Natenzon, Axel Niemeyer, Luciano Pomatto, Justus Preusser, Collin Raymond, Kota Saito, Fedor Sandomirskiy, Bumin Yenmez, Kemal Yildiz, and Xinhan Zhang as well as seminar participants at Bilkent, BRIC 11, BSE Summer Forum Workshop on Choice and Decision, EC'25, Italian Junior Workshop on Economic Theory 2025, Liverpool, RUD 2025, UC Davis, and WUSTL for discussions and helpful comments. This paper subsumes ``Identifying Restrictions on the Random Utility Model".}}
\author{Peter Caradonna\thanks{Division of the Humanities and Social Sciences, Caltech.  Email: \href{mailto:ppc@caltech.edu}{\nolinkurl{ppc@caltech.edu}}.} \hspace{.1cm} and Christopher Turansick\thanks{Department of Decision Sciences and IGIER, Bocconi University. Email: \href{mailto:christopher.turansick@unibocconi.it}{\nolinkurl{christopher.turansick@unibocconi.it}}.
}}
\date{\today}

\begin{document}
\onehalfspacing

\maketitle

\begin{abstract}
    We characterize the identified sets of a wide range of stochastic choice models, including random utility, various models of boundedly-rational behavior, and dynamic discrete choice. In each of these settings, we show two distributions over choice rules are observationally equivalent if and only if they can be obtained from one another via a finite sequence of simple swapping transforms. We leverage this to obtain complete descriptions of both the defining inequalities and extreme points of these identified sets. In cases where choice frequencies vary smoothly with some parameters, we provide a novel global-inverse result for practically testing identification.
\end{abstract}

\section{Introduction}\label{sec:Intro}

A fundamental question for any model of behavior is to what degree its underlying primitives can be uniquely determined from individual or aggregate-level data. Even when exact identification fails, knowledge of the identified set can provide crucial information for the design of policy, counterfactual simulations, and welfare analysis.\footnote{E.g.\ \cite{manski2007partial, agarwal2018demand, kalouptsidi2021counterfactual}. See also  \cite{marschak1953economic} and  \cite{allen2019identification}.}\medskip

Despite their practical importance, the identification properties of many workhorse models of stochastic choice, including the random utility model, have remained poorly understood.\footnote{See, e.g., \cite{fishburn1998stochastic, mcclellon2015unique, strzalecki2025stochastic}.} In response, an extensive literature has emerged, studying additional restrictions under which analysts can uniquely recover an underlying distribution of parameters governing behavior from choice share data.\footnote{E.g.\ \cite{luce1959individual, mcfadden1972conditional, gul2006random, apesteguia2017single, manzini2018dual, filiz2023progressive, yang2023random, suleymanov2024branching}. Similar questions arise in the context of demand invertibility, see \cite{berry2024nonparametric}.}\medskip

In this paper, we provide complete characterizations of the identified sets of a wide range of stochastic choice theories, under arbitrary supplemental restrictions.  Our primary observation is that, in each of these settings, each set of empirically indistinguishable distributions (e.g.\ over preferences) can be straightforwardly described in terms of a simple family of transformations. We term these \emph{Ryser swaps}, in light of a connection to the discrete tomography literature.\footnote{E.g.\ \cite{ryser1957combinatorial, fishburn1991sets, Kong1999}. These tools have found recent application in economic theory, see  \cite{he2025private}.}  In particular, we show two distributions generate identical choice shares over any menu if, and only if, they are related by a finite sequence of such swaps.\medskip

For expositional clarity, we focus our discussion on the random utility model. However our findings apply straightforwardly to many other settings, including to models of incomplete or boundedly-rational behavior, correlated decision-making, and dynamic stochastic choice.\footnote{We discuss these extensions in \autoref{sec:Extensions}. See also \hyperlink{bru}{Online Appendix B} for further details.}

\begin{example}\label{introex}

A population of voters are faced with four potential urban planning policies.  Policy $a$ calls for regional high-speed rail expansion, and $b$ for large-scale investment in local bus and subway systems, whereas $c$ for moderate congestion-pricing, and $d$ an expansion of bike lanes and pedestrian zones.\medskip

The population is comprised of four types of voters. The first are `big picture' activists, who prefer large-scale investment in public transit and a centralized policies ($a \succ b \succ c \succ d$), whereas incrementalists prefer more funding to less, but prize more immediate, practical transit options, and a ground-up approach to policy ($b \succ a \succ d \succ c$). Similarly, suburban commuters prefer rail to buses and dislike taxes on their vehicles ($a \succ b \succ d \succ c$), whereas urban commuters primarily use public transit and prioritize less congested streets ($b \succ a \succ c \succ d$). We summarize these four preferences below.\footnote{The observation that a distribution with full support on these four preferences is unidentified from choice frequency data is due to \cite{fishburn1998stochastic}. Moreover, such examples are minimal: the random utility model is known to be identified whenever there are three or fewer alternatives, e.g.\ \cite{block1959random}.}\medskip

    \begin{table}[h!]
\centering
 \begin{tabular}{c c c c} 
 
 $\succ_1$ & $\succ_2$ & $\succ_3$ & $\succ_4 $\\ [0.25ex] 
 \hline\hline
  a & b & a & b\\ 
  b & a & b & a\\
  c & d & d & c\\
  d & c & c & d\\ 
 \end{tabular}
\end{table}

Suppose straw-poll data over various subsets of these policies is consistent with the societal distribution of preferences being $\mu = \big(\frac{1}{4}, \frac{1}{4}, \frac{3}{8}, \frac{1}{8}\big)$, and consider now a new policy, likely to benefit suburban commuters.\footnote{That is, voters who hold $\succ_3$.} What can be said about the size of this share of the population?\medskip

Our results show suburban commuters constitute anywhere between $\frac{1}{4}$ and $\frac{5}{8}$ of the population. Moreover, these bounds are tight: for any fraction $\frac{1}{4} \le \alpha \le \frac{5}{8}$, there is a distribution over these four preferences where the population share of suburban commuters is precisely $\alpha$, and which yields identical vote shares to $\mu$ over \emph{any} menu of policies.\medskip

The key feature in this example is that all four preferences rank $a$ and $b$ above $c$ and $d$, but disagree about the relative rankings of the alternatives within these groupings.  In particular, $\succ_3$ and $\succ_4$ can be obtained by taking $\succ_1$ and $\succ_2$, and exchanging their respective rankings between $c$ and $d$, while preserving all other comparisons.\medskip

Given any distribution over these four preferences with full support, by subtracting some small, but equal amount of probability mass from either $\succ_1$ and $\succ_2$, or $\succ_3$ and $\succ_4$, and re-assigning it equally to the opposing pair (here, obtained by exchanging only the rankings of $c$ and $d$) one obtains a modified distribution with identical choice probabilities on all menus. Intuitively, such transformations affect only the \emph{correlation} between preferring $ a$ to  $b$ and $c$ to  $d$, but not the likelihood of these events. \medskip

In this setting, it turns out these are the \emph{only} transformations which preserve all choice probabilities. As such, the minimal and maximal possible population frequencies for $\succ_3$ can be obtained by simply maximizing the amount of mass uniformly transferred from the pair $\{\succ_3,\succ_4\}$ to ~$\{\succ_1, \succ_2\}$~ or vice-versa, yielding the extremal, empirically indistinguishable distributions $\underbar{$\mu$} = \big(\frac{3}{8}, \frac{3}{8}, \frac{1}{4}, 0\big)$ and $\bar{\mu} = \big(0, 0, \frac{5}{8}, \frac{3}{8}\big)$.\hfill $\blacksquare$
\end{example}

As illustrated by \autoref{introex}, failures of identification arise for a straightforward reason: choice probabilities do not, in general, uniquely determine the \emph{correlation} between choices.\footnote{For instance, in \autoref{introex}, the correlation between choices on menus $\{a,b\}$ and $\{c,d\}$.}  Indeed, the identified sets of the random utility model are spanned precisely by the possible correlation patterns that cannot be ruled out by an empiricist with access only to aggregate, rather than individual, level data.\medskip

Our first main result characterizes the full collection of preference distributions consistent with a given dataset. We show any transformation of a distribution which preserves all choice probabilities is equivalent to applying a finite sequence of simple swaps of the form considered in \autoref{introex}. Using this representation, we provide a straightforward, geometric description of the random utility model's identified sets, both in terms of their defining inequalities, and their extreme points. As consumer welfare measures typically rely on preference information (\citealt{deaton1980measurement, mckenzie1983measuring}), this can be used, e.g., to obtain tight bounds on the estimated aggregate or distributional effects of policy interventions.\footnote{See, e.g., \citet{deb2023revealed}. \cite{tebaldi2023nonparametric} and \cite{kamat2025estimating} consider similar questions in the context of attribute variation. Our results may be seen as providing complementary tools for settings with menu variation, e.g.\ \cite{mcfadden2005revealed,   kitamura2018nonparametric}, as well as to dynamic, or imperfectly rational, models of behavior; see \autoref{sec:Extensions}.}\medskip

We then investigate how our results specialize under common, structured classes of supplemental restrictions. For example, in many instances an analyst may wish to a priori restrict the \emph{set} of preferences assumed held by a population, but not their relative frequencies.\footnote{
For example, expected utility preferences on domains of lotteries (\citealt{gul2006random}, \citealt{apesteguia2018monotone}), exponentially discounted or quasi-hyperbolic preferences over streams (\citealt{lu2018random}), subjective expected utility over acts (\citealt{lu2016random, duraj2018dynamic}), or single-peaked preferences on ordered domains (\citealt{apesteguia2017single}). Such restrictions are also implicit in more general, dynamic settings, e.g.\ \cite{rust1987,lu2024did}; see also \cite{frick2019dynamic}.}  We characterize identified sets subject to arbitrary support restrictions, as well as those families of preferences $\mathcal{S}$ with the property that any data is consistent with at most a single measure whose support belongs to $\mathcal{S}$. In practice, $\mathcal{S}$ often consists of restrictions of some parametric family of utilities.\footnote{E.g.\ \cite{apesteguia2018monotone} consider the case of random CARA or CRRA expected utility.} In this case, our results provide an exact test for practitioners to determine when choice probabilities uniquely identify a population distribution over parameters.\medskip  

We also consider cases in which choice probabilities are specified directly, as smooth functions of some vector of parameters.  Using topological techniques, we provide a pair of necessary and sufficient conditions for the global invertibility of general, non-linear systems. In contrast to many existing results, we require no monotonicity properties.\footnote{C.f.\ \cite{gale1965jacobian, beckert2004invertibility, simsek2005uniqueness, berry2013connected, lindenlaub2017sorting, wang2021blp, allen2022injectivity}. See also `dominant diagonal' conditions, e.g.\ \cite{mckenzie1960matrices, rosen1965existence}.} Thus, for example, applied to the problem of determining the invertibility of a smooth demand, our results are applicable even when goods may fail to be substitutes, and in settings where the law of demand may fail. We demonstrate the practical applicability of our conditions by proving a novel result, that the dynamic, logit habit formation model of \cite{TuransickIntertemporal2025} is not identified from unconditional choice data.\medskip

Finally, when the set of alternatives is ordered, an elegant approach due to \cite{apesteguia2017single} shows that requiring single-crossing structure on the support of a rationalization yields identification.\footnote{See also \cite{filiz2023progressive, apesteguia2023random}.} However, such rationalizations need not always exist, even for data compatible with the random utility model. Using our notion of a Ryser swap, we provide a principled generalization we term \emph{swap progressivity}.  Whenever the data admit a single-crossing representation, this will always be the unique swap-progressive rationalization. However, we show a unique swap-progressive representation exists for \emph{any} data compatible with the random utility model. This provides a novel, practical identification strategy that incorporates the natural ordering over alternatives.\medskip

The paper proceeds as follows. In \autoref{sec:Model}, we formally introduce the random utility model. In \autoref{sec:GeometryResults}, we formalize our core notion of a Ryser swap, and present our main geometric characterization. \autoref{sec:Support} provides a dual characterization of identified sets in terms of their extreme points, and characterizes identifying support restrictions. In \autoref{sec:Ordered}, we introduce our notion of swap progressivity and establish its properties, and \autoref{sec:Parm} considers the identification of general, parametric stochastic choice models. Finally, \autoref{sec:Extensions} discusses extensions to other environments.  To more concretely relate our contributions to existing work, we defer our discussion of related literature to \autoref{sec:RelatedLit}. \autoref{sec:Conc} concludes.

\section{The Random Utility Model}\label{sec:Model}
Let $X$ denote a fixed, finite set of alternatives over which an individual chooses.  A {\bf preference} is a linear order on $X$; we denote the set of all preferences by $\mathcal{L}$.\footnote{A linear order is a complete, transitive, and antisymmetric binary relation on $X$.} To conserve on notation, we will often write preferences as strings of alternatives, arranged in descending order, i.e.\ `$abcd$' denotes the ranking $a \succ b \succ c \succ d$. For any preference $\succ$ and any $0 \le k \le \vert X \vert$, the $k$-{\bf initial} and $k$-{\bf terminal} segments of $\succ$ are the ordered strings consisting of the $k$ most-preferred alternatives and the $\vert X \vert - k$ least-preferred.\footnote{For $k=0$ (resp.\ $k = \vert X \vert$) we define $s^\uparrow_k(\succ)$ (resp.\ $s^k_\downarrow(\succ)$) as the empty string.} We will denote these by $s^\uparrow_k(\succ)$ and $s^k_\downarrow(\succ)$ respectively. For example, if $\succ$ corresponds to the ordering $abcde$, then:
\[
    s^\uparrow_2(\succ) = ab \quad \textrm{ and } \quad s^2_\downarrow(\succ) = cde.
\]

A function $\rho: X \times 2^X \setminus \{\varnothing\} \to [0,1]$ defines a {\bf random choice rule} if, for all non-empty subsets $A \subseteq X$,
\[
    \sum_{x \in A} \rho(x,A) = 1.
\]
Random choice rules are assumed observable; they constitute the basic data in our identification problem.  For any finite set $A$, we use $\Delta(A)$ to denote the set of probability measures over $A$. A {\bf restriction} of the random utility model is any subset $\mathcal{M} \subseteq \Delta(\mathcal{L})$. We say a random choice rule is rationalizable subject to a restriction $\mathcal{M}$ if there exists some probability distribution $\mu \in \mathcal{M}$ such that, for all $x \in A \subseteq X$, we have:
\[
    \rho(x,A) = \mu\{\succ \, \in \mathcal{L}\;  \vert \; x \textrm{ is maximal in } A \} = \sum_{\succ \in \mathcal{L}} \mu(\succ) \mathbbm{1}_{\{x \; \succ \,A \setminus x\}}.
\]
Under the interpretation of a restriction $\mathcal{M}$ as describing  a set of possible makeups of a heterogeneous society, a collection of observed choice frequencies $\rho$ are $\mathcal{M}$-rationalizable if and only if they arise as the distribution of constrained-optimal outcomes according to some population composition $\mu \in \mathcal{M}$.\medskip

We say two measures $\mu, \nu \in \Delta(\mathcal{L})$ are {\bf observationally equivalent} if, for all $x \in A \subseteq X$,
\begin{equation}\label{behavioralequiv}
    \mu\{\succ \, \in \mathcal{L}\;  \vert \; x \textrm{ is maximal in } A \} = \nu\{\succ \, \in \mathcal{L}\;  \vert \; x \textrm{ is maximal in } A \},
\end{equation}
i.e. they generate identical choice frequencies on every choice set $\varnothing \subsetneq A \subseteq X$.  Finally, a restriction $\mathcal{M}$ is  {\bf identifying} if it contains no pair of distinct, observationally equivalent measures.

\section{The Geometry of Identification}\label{sec:GeometryResults}

The unrestricted random utility model (i.e.\ $\mathcal{M} = \Delta(\mathcal{L})$) has long been known to be unidentified.\footnote{See \cite{fishburn1998stochastic}.} Thus, to each distribution $\mu \in \Delta(\mathcal{L})$, there corresponds some generally non-singleton collection of observationally equivalent distributions. \medskip

Let $\mathcal{P}$ denote the set of random choice rules on $X$, and define the mapping $\Phi: \Delta(\mathcal{L}) \to \mathcal{P}$ via:
\begin{equation}\label{measuretochoiceprobmap}
    \Phi(\mu)_{(x,A)} = \begin{cases}
         \mu\{\succ \, \in \mathcal{L}\;  \vert \; x \textrm{ is maximal in } A \} & \textrm{ if } x \in A\\
         0 & \textrm{ if } x \not \in A.
    \end{cases}
\end{equation}
By \eqref{behavioralequiv} two distributions $\mu,\nu \in \Delta(\mathcal{L})$ are observationally equivalent if and only if $\Phi(\mu) = \Phi(\nu)$.  However, as $\Phi$ is linear in $\mu$, there exists some subspace $\mathcal{R} \subseteq \mathbb{R}^\mathcal{L}$ such that, whenever $\Phi(\mu) = \rho$, the identified set $\Phi^{-1}(\rho)$ is precisely $\big(\mu + \mathcal{R} \big) \cap \Delta(\mathcal{L})$.\footnote{Here we identify $\mathbb{R}^\mathcal{L}$ with the space of all signed measures over linear orders on $X$.}  In particular, every identified set is fully characterized by this subspace.

\subsection{Compatible Pairs and the Ryser Subspace}

We say that a pair of preferences are $k$-{\bf compatible}, $0 \le k \le \vert X \vert$, if both preferences agree on the set of $k$-best alternatives.\footnote{In particular, any pair of preferences is vacuously $0$-compatible.} In other words, $\succ_1$ and $\succ_2$ are $k$-compatible if $s^\uparrow_k(\succ_1)$ can be obtained by permuting the terms of the sequence $s^\uparrow_k(\succ_2)$, and analogously $s^k_\downarrow(\succ_1)$ from $s^k_\downarrow(\succ_2)$. This ensures the concatenated sequences:
\begin{equation} \label{swap}
    \succ_1' \,=\, s^\uparrow_k(\succ_1) \, s^k_\downarrow(\succ_2) \quad \textrm{ and } \quad \succ_2' \,= \,s^\uparrow_k(\succ_2) \, s^k_\downarrow(\succ_1)
\end{equation}
define valid preferences. We refer to the pair of preferences obtained in this manner as the $k$-{\bf conjugates} of $\succ_1$ and $\succ_2$. If $\succ_1$ and $\succ_2$ disagree on both (i) the ranking of the $k$-best alternatives and (ii) the ranking of the $(\vert X \vert-k)$-worst, $\succ_1'$ and $\succ_2'$ will be distinct from either $\succ_1$ or $\succ_2$; in this case we refer to $\succ_1$ and $\succ_2$ as {\bf non-trivially} $k$-compatible.\medskip

Our interest in compatible pairs of preferences is motivated by the following proposition, which says that any uniform distribution supported on a compatible pair of preferences is behaviorally equivalent to the uniform distribution over the pair obtained via \eqref{swap}.

\begin{proposition}\label{conjugatelemma}
    Let $\succ_1, \succ_2 \in \mathcal{L}$ be $k$-compatible, with $k$-conjugates $\succ_1'$ and $\succ_2'$. Then the uniform distributions on $\{\succ_1, \succ_2\}$ and $\{\succ_1', \succ_2'\}$ are observationally equivalent.
\end{proposition}

As the following example illustrates, the observational indeterminacy highlighted by \autoref{conjugatelemma} is a consequence of the fact that a joint distribution is not, in general, uniquely determined by the requirement that it possess uniform marginals.

\begin{example}\label{fishex}
Consider again \autoref{introex}, where $X = \{a,b,c,d\}$, and there are four preferences: $\succ_1: abcd$, $\succ_2: badc$, $\succ_3: abdc$, and $\succ_4: bacd$. Note that the pair $\succ_1$ and $\succ_2$ are $2$-compatible, with $2$-conjugates $\succ_3$ and $\succ_4$.\footnote{That is, $\succ_3 \, =\,  s^\uparrow_2(\succ_1)\, s^2_\downarrow(\succ_2)$, and $\succ_4 \, = \, s^\uparrow_2(\succ_2)\, s^2_\downarrow(\succ_1)$.} We may view these four preferences as elements of a product set $\{ab, ba\} \times \{cd, dc\}$, by associating an ordered pair of strings with their concatenation, e.g.:
\[
    (ab, cd) \mapsto abcd.
\]
A joint distribution over $\{ab, ba\} \times \{cd, dc\}$ generates the same choice probabilities on every menu as the uniform distribution on $\{\succ_1, \succ_2\}$ if and only if both of its marginals are themselves uniform.\footnote{This can be verified by direct computation or, e.g., follows from Theorem 5 of \cite{falmagne1978representation}.} Thus, in this setting, \autoref{conjugatelemma} asserts the observational equivalence of the two extremal joint measures with uniform marginals: $\mu_{12}$, the uniform on $\{\succ_1,\succ_2\}$, which perfectly correlates the choice of $a$ from $\{a,b\}$ with $c$ from $\{c,d\}$, and $\mu_{34}$, the uniform on $\{\succ_3,\succ_4\}$ which perfectly anti-correlates these choices. See \autoref{fig:ProdAndSimp}.\hfill $\blacksquare$
\end{example}

\begin{figure}
\centering
\begin{subfigure}[t]{.45\textwidth}
\begin{tikzpicture}[scale = .65]
        \draw[thick] (0,0)--(0,6)--(6,6)--(6,0)--cycle;
        \draw[thick] (0,0)--(6,0)--(6,3)--(0,3)--cycle;
        \draw[thick] (0,0)--(0,6)--(3,6)--(3,0)--cycle;

        \node[left] at (-.5,1.5) {  $ba$};
        \node[left] at (-.5,4.5) {  $ab$};

        \node[left] at (-1.75,3) { \large $\,$};

        \node[above] at (1.5, 6.5) {  $cd$};
        \node[above] at (4.5, 6.5) {  $dc$};

        \node[above] at (3, 7.5) { \large $\,$};

        \node at (1.5,4.5) { $\succ_1$};
        \node at (1.5,1.5) { $\succ_4$};
        \node at (4.5,4.5) { $\succ_3$};
        \node at (4.5,1.5) { $\succ_2$};

    \end{tikzpicture}
        \subcaption{The preferences of \autoref{introex}, viewed as a product set. The uniforms $\mu_{12}$ and $\mu_{34}$ differ only in their correlation structure.}
\end{subfigure}
\begin{subfigure}[t]{.45\textwidth}
\begin{tikzpicture}[line join = round, line cap = round, scale =2.5]

\node at (-1.5,0) {};
\coordinate [label=above:$\succ_4$] (3) at (0,{sqrt(2)},0);
\coordinate [label=left:$\succ_3$] (2) at ({-.5*sqrt(3)},0,-.5);
\coordinate [label=below:$\succ_2$] (1) at (0,0,1);
\coordinate [label=below right:$\succ_1$] (0) at ({.5*sqrt(3)},0,-.5);

\coordinate [label= above left: \footnotesize $\mu_{34}$] (4) at ({-.25*sqrt(3)}, {.5*sqrt(2)}, -0.25);
\coordinate [label = below right: \footnotesize $\mu_{12}$] (5) at ({.25*sqrt(3)}, 0, 0.25);
\coordinate (6) at (0, {.25*sqrt(2)}, 0);

\draw[dashed, thick] (0)--(2);
\draw[very thick, Orange] (4) -- (5);
\draw[very thick, dashed, Orange] ($(4)!0.6!(0)$)--($(5)!0.6!(0)$);
\draw[very thick,dashed, Orange] ($(4)!0.6!(1)$)--($(5)!0.6!(1)$);

\draw[thick] (1)--(0)--(3)--cycle;
\draw[thick] (2)--(1)--(3)--cycle;
\draw[thick] (1)--(0);
\draw[thick] (1)--(2);
\draw[thick] (2)--(3);
\draw[thick] (1)--(3);
\draw[thick] (0)--(3);

\draw[dotted] (4)--(1);
\draw[dotted] (4)--(0);

\draw[fill, fill opacity = .05, dotted] (4)--(1)--(0)--cycle;

\fill (0) circle [black!90, radius=.5pt];
\fill (1) circle [black!90, radius=.5pt];
\fill (2) circle [black!90, radius=.5pt];
\fill (3) circle [black!90, radius=.5pt];
\fill (4) circle [black!90, radius=.5pt];
\fill (5) circle [black!90, radius=.5pt];

\draw[fill, black!90, radius=.35pt] ($(4)!0.6!(0)$) circle;
\draw[fill, black!90, radius=.35pt] ($(5)!0.6!(0)$) circle;
\draw[fill, black!90, radius=.35pt] ($(4)!0.6!(1)$) circle;
\draw[fill, black!90, radius=.35pt] ($(5)!0.6!(1)$) circle;

\filldraw[Orange!90,radius=.6pt] (4) circle;
\filldraw[Orange!90,radius=.6pt] (5) circle;


\end{tikzpicture}
\subcaption{\autoref{conjugatelemma} asserts the observational equivalence of the uniform distributions $\mu_{12}$ and $\mu_{34}$.}
\end{subfigure}
\caption{By \autoref{conjugatelemma}, $\mu_{12}$ and $\mu_{34}$ are observationally equivalent. By the linearity of \eqref{measuretochoiceprobmap}, both are also equivalent to any distribution in their convex hull, and any pair of distributions belonging to some common translation of this set (e.g.\ the dashed segments) are observationally equivalent.}
\label{fig:ProdAndSimp}
\end{figure}
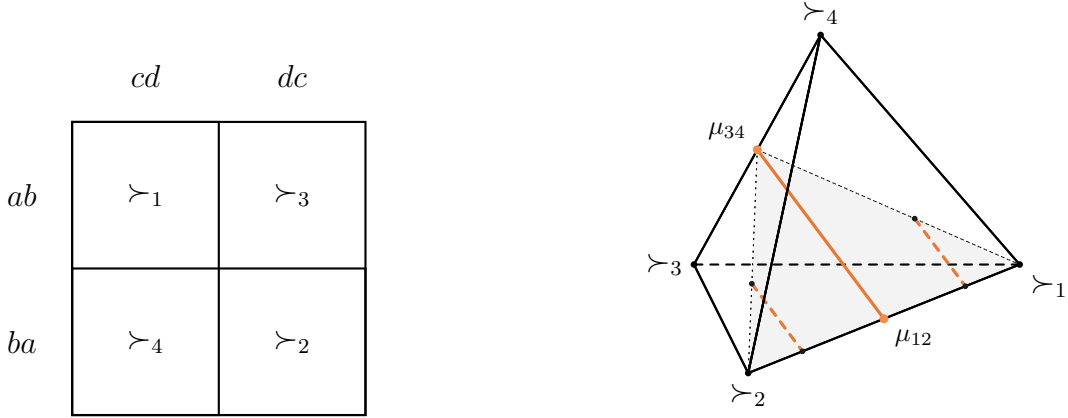

 We say a signed measure $R \in \mathbb{R}^\mathcal{L}$ defines a {\bf Ryser swap} if:
\[
    R = \mathbbm{1}_{\{\succ_1',\succ_2'\}} - \mathbbm{1}_{\{\succ_1,\succ_2\}},
\]
where $\succ_1$ and $\succ_2$ are $k$-compatible, and $\succ_1'$ and $\succ_2'$ are $k$-conjugate for some $k$.\footnote{Note $\succ_1$ and $\succ_2$ are trivially $k$-compatible if and only if $R = 0$.} A signed measure is a {\bf weighted Ryser swap} if it is proportional to a Ryser swap. Adding a weighted Ryser swap to some initial distribution $\mu$ may be regarded as a compound transformation, which first separates out from $\mu$ a product measure, applies to it a marginal-preserving transformation, then recombines the result with the remaining mass from $\mu$ to obtain a modified measure $\mu'$.  As in \autoref{introex}, whenever $\mu'$ is a probability distribution, it differs from $\mu$ only in the \emph{correlation} between choices, not their frequencies.\medskip

Finally, we define the {\bf Ryser subspace} $\mathcal{R}\subseteq \mathbb{R}^\mathcal{L}$ as the span of the Ryser swaps:
\[
    \mathcal{R} = \textrm{span}\big\{R \in \mathbb{R}^\mathcal{L} \; \vert \; R \textrm{ is a Ryser swap}\big\}.
\]
Every signed measure in $\mathcal{R}$ corresponds to a transformation which applies a finite sequence of weighted swaps.  By linearity of $\Phi$, given $\bar{R} \in \mathcal{R}$ and $\mu \in \Delta(\mathcal{L})$, if $\mu + \bar{R}$ defines a probability distribution, \autoref{conjugatelemma} implies it is observationally equivalent to $\mu$.\footnote{\autoref{conjugatelemma} asserts that, given compatible $\succ_1,\succ_2$ with conjugates $\succ_1', \succ_2'$, the uniforms $\mu_{12}$ and $\mu_{1'2'}$ are observationally equivalent, and hence $\Phi(\mu_{12}) = \Phi(\mu_{1'2'})$. By linearity:
\[
    \Phi\big(\mu + \varepsilon(\mu_{1'2'}-\mu_{12})\big) = \Phi(\mu) + \varepsilon \, ( 0) = \Phi(\mu).
\]
Since every $\bar{R} \in \mathcal{R}$ is a weighted sum of Ryser swaps, $\mu$ and $\mu + \bar{R}$ are observationally equivalent.} Our next theorem shows that, in fact, \emph{every} pair of observationally equivalent distributions are related in this manner.

\begin{theorem}\label{ryserspacetheorem}
    Two distributions $\mu, \mu' \in \Delta(\mathcal{L})$ are observationally equivalent if and only if $\mu - \mu' \in \mathcal{R}$.  In particular, for any $\mu \in \mathcal{M} \subseteq \Delta(\mathcal{L})$, the set of observationally equivalent distributions in $\mathcal{M}$ is precisely $\big(\mu + \mathcal{R}\big) \cap \mathcal{M}$.
\end{theorem}

\autoref{ryserspacetheorem} shows that not only do Ryser swaps preserve choice probabilities, they in fact generate the set of \emph{all} transformations which do so. In other words, two distributions over preferences are observationally equivalent if and only if the first can be generated from the second (or vice-versa) by a finite sequence of such swaps. Geometrically, this implies the set of distributions observationally equivalent to some $\mu \in \Delta(\mathcal{L})$ is simply the intersection of the affine subspace $\mu + \mathcal{R}$ with $\Delta(\mathcal{L})$.\medskip

This provides a practical means of computing features of the identified set. Since $\mathcal{R}$ does not depend on the data, once a choice of $X$ is specified, a basis $R_1, \ldots, R_N$ for $\mathcal{R}$ can be obtained from the set of Ryser swaps. Then, given any rationalization $\mu \in \Delta(\mathcal{L})$, optimal bounds on, e.g., the sizes of different shares of the population across all rationalizations, can be computed by evaluating a straightforward linear program over the weights $\alpha_1,\ldots,\alpha_N$ assigned to these basis elements, subject only to the constraint $\mu + \sum_i \alpha_i R_i \ge \mathbf{0}$.\footnote{One could evaluate a similar linear program over $\Delta(\mathcal{L})$ using $\Phi(\cdot) = \rho$ as a constraint, without appealing to \autoref{ryserspacetheorem} (e.g.\ \citealt{deb2023revealed}). However, doing so requires requires $\vert X \vert !-1$ variables and hence may become computationally difficult in larger environments. In contrast, identified sets are generally far lower dimensional; in \autoref{introex}, $\Delta(\mathcal{L})$ is 23-dimensional whereas the identified set was 1-dimensional. Thus the explicit description provided by \autoref{ryserspacetheorem} may allow for significant computational savings in practice.}\medskip

To illustrate, consider again \autoref{introex}. As there are only four alternatives, any non-trivially compatible pair must agree on the best (and worst) two alternatives, but not their ordering. As a consequence, if the measure $\mu + R$ is a probability distribution for some $R \in \mathcal{R}$, $R$ must subtract mass only from $\succ_1,\ldots, \succ_4$, and reassign it to preferences formed by matching terminal segments $s^2_\downarrow \in \{cd, dc\}$ to initial segments $s_2^\uparrow \in \{ab,ba\}$. But $\succ_1,\ldots, \succ_4$ are the \emph{only} preferences formed this way. It follows $R$ must be proportional to the sole (up to sign) non-trivial Ryser swap supported on these preferences, $\bar{R} = \mathbbm{1}_{\{\succ_1,\succ_2\}} - \mathbbm{1}_{\{\succ_3,\succ_4\}}$. In particular, this implies the orange line segments in \autoref{fig:ProdAndSimp} are precisely the identified sets.\footnote{Indeed, given some rationalization $\mu$, tight bounds on \emph{any} property of the identified set can be computed simply by optimizing $\mu + \alpha \bar{R}$ over the scalar $\alpha$, subject to the constraint this vector belong to the simplex.}

\begin{example}\label{sixaltexample}
In more complex environments, elements of $\mathcal{R}$ may have richer form. Suppose $X = \{a,b,c,d,e,f\}$, and consider the following six preferences.

\begin{table}[h!]
\small
\centering
 \begin{tabular}{c c c | c c c} 
 
 $\succ_1$ & $\succ_2$ & $\succ_3$ & $\succ_1' $ & $\succ_2'$ & $\succ_3'$\\ [0.25ex] 
 \hline\hline
  a & b & c & a & b & c\\ 
  b & a & d & b & a & d\\
  c & e & b & e & c & b\\
  d & f & a & f & d & a\\ 
  e & c & f & c & f & e\\
  f & d & e & d & e & f
 \end{tabular}
\end{table}

\noindent Here, $\mu_{123}$, the uniform distribution on $\{\succ_1,\succ_2, \succ_3\}$, is observationally equivalent to $\mu_{123}'$, the uniform on $\{\succ_1', \succ_2', \succ_3'\}$. However, unlike in \autoref{introex}, these distributions do not differ by any single weighted Ryser swap.\footnote{To see this, note any Ryser swap between any pair of these six preferences is either trivial or places non-zero mass on some preference outside the set.} Nevertheless, they do so by a \emph{sequence} of such swaps.

\[\small
    \begin{array}{ccc}
    \succ_1 & \succ_2 & \succ_3\\
    \hline \hline
    \textcolor{NavyBlue}{\textrm{a}} & \textcolor{NavyBlue}{\textrm{b}} & \textrm{c}\\
    \textcolor{NavyBlue}{\textrm{b}} & \textcolor{NavyBlue}{\textrm{a}} & \textrm{d}\\ \cline{1-2}
    \textcolor{NavyBlue}{\textrm{c}} & \textcolor{NavyBlue}{\textrm{e}} & \textrm{b}\\
    \textcolor{NavyBlue}{\textrm{d}} & \textcolor{NavyBlue}{\textrm{f}} & \textrm{a}\\
    \textcolor{NavyBlue}{\textrm{e}} & \textcolor{NavyBlue}{\textrm{c}} & \textrm{f}\\
    \textcolor{NavyBlue}{\textrm{f}} & \textcolor{NavyBlue}{\textrm{d}} & \textrm{e}\\
    \end{array}
    \longrightarrow 
    \begin{array}{ccc}
    \succ_1' & \succ & \succ_3\\
    \hline \hline
    \textrm{a} & \textcolor{Orange}{\textrm{b}} & \textcolor{Orange}{\textrm{c}}\\
    \textrm{b} & \textcolor{Orange}{\textrm{a}} & \textcolor{Orange}{\textrm{d}}\\
    \textrm{e} & \textcolor{Orange}{\textrm{c}} & \textcolor{Orange}{\textrm{b}}\\
    \textrm{f} & \textcolor{Orange}{\textrm{d}} & \textcolor{Orange}{\textrm{a}}\\\cline{2-3}
    \textrm{c} & \textcolor{Orange}{\textrm{e}} & \textcolor{Orange}{\textrm{f}}\\
    \textrm{d} & \textcolor{Orange}{\textrm{f}} & \textcolor{Orange}{\textrm{e}}\\
    \end{array}
    \longrightarrow
    \begin{array}{ccc}
    \succ_1' & \succ_2' & \succ_3'\\
    \hline \hline
    \textrm{a} & \textrm{b} & \textrm{c}\\
    \textrm{b} & \textrm{a} & \textrm{d}\\
    \textrm{e} & \textrm{c} & \textrm{b}\\
    \textrm{f} & \textrm{d} & \textrm{a}\\
    \textrm{c} & \textrm{f} & \textrm{e}\\
    \textrm{d} & \textrm{e} & \textrm{f}\\
    \end{array}
\]

\noindent To construct such a sequence, first transfer all mass from $\succ_1$ and $\succ_2$ under $\mu_{123}$ to their $2$-conjugates (the pair formed by swapping their terminal segments $cdef$ and $efcd$). Label these resulting preferences $\succ_1'$ and $\succ$. Then, transfer all mass on $\succ$ and $\succ_3$ to their $4$-conjugates (the pair obtained from exchanging $ef$ and $fe$) to yield $\mu_{123}'$.\hfill $\blacksquare$ 
\end{example}

\section{Support Restrictions and Extreme Points}\label{sec:Support}

Perhaps the simplest class of restrictions are those in which an analyst constrains the set of preferences held in a population. We say $\mathcal{M} \subseteq \Delta(\mathcal{L})$ defines a {\bf support restriction} if:
\[
    \mathcal{M} = \big\{ \mu \in \Delta(\mathcal{L}) : \textrm{supp}(\mu) \subseteq \mathcal{S}\big\},
\]
for some set of preferences $\mathcal{S} \subseteq \mathcal{L}$. Such restrictions are natural when a modeler wishes to a priori restrict the set of preferences, but not their relative frequencies.\medskip

\begin{example}\label{ex:MonotoneEU}
Let $X$ consist of the following four monetary lotteries:
    \begin{equation*}
        \begin{split}
            a =\frac{1}{2}\big[\$10\big]+\frac{1}{2}\big[\$20\big] &\text{, } \quad b= \frac{1}{2}\big[\$5\big] + \frac{1}{2}\big[\$25\big]  \\
            c= \frac{1}{4}\big[\$10\big] + \frac{3}{4}\big[\$20\big] &\text{, } \quad d = \frac{1}{4}\big[\$5\big] + \frac{3}{4}\big[\$25\big].
        \end{split}
    \end{equation*}
    Suppose we wish to consider only those preferences $\mathcal{S} \subseteq \mathcal{L}$ which are consistent with some (monotone) expected utility representation. To determine this subset, note $c \succsim_{FOSD} a$ and $d \succsim_{FOSD} b$, hence any preference in $\mathcal{S}$ ranks $c \succ a$ and $d \succ b$.\footnote{Here, $\succsim_{FOSD}$ corresponds to first order stochastic dominance.} Moreover, 
    \[
        c = \frac{1}{2} \big[a \big] + \frac{1}{2} \big[\$20\big] \quad \textrm{ and } \quad d= \frac{1}{2}\big[b\big] + \frac{1}{2}\big[\$25\big],
    \]
    thus every preference in $\mathcal{S}$ which ranks $b \succ a$, also obeys $d \succ c$. These are all the restrictions imposed by monotonicity and the expected utility axioms, hence $\mathcal{S}$ consists of the five orders:
     \begin{table}[H]
\centering
 \begin{tabular}{c c c c c} 
 
 $\succ_1$ & $\succ_2$ & $\succ_3$ & $\succ_4 $ & $\succ_5$ \\ [0.25ex] 
 \hline\hline
    d & d  & d & c & c \\
    c & c  & b & d & a \\
    b & a  & c & a & d \\
    a & b  & a & b & b
 \end{tabular}
\end{table}
    \noindent listed here.\hfill $\blacksquare$
\end{example}

While \autoref{ex:MonotoneEU} highlights the case of expected utility preferences, one could likewise consider other restrictions, from parametric sub-families such as CARA or CRRA, to broader classes of risk preferences, such as betweenness or disappointment aversion. Similarly, many other environments suggest natural restrictions, such as exponentially discounted preferences over consumption streams, or subjective expected utility preferences over acts.\medskip  

The question we now consider is how such restrictions affect the set of rationalizing distributions, and when the probability weights on these and other families of orders can be uniquely recovered from choice frequency data alone.\medskip

Let $\succ_1,\ldots, \succ_N$ be a (finite) sequence of preferences, allowing repetition. For each $1 \le k \le \vert X \vert$, let $\Pi_k$ denote the partition of $\{1,\ldots, N\}$ such that $i$ and $j$ belong to the same partition element if and only if $\succ_i$ and $\succ_j$ are $(k-1)$-compatible.\footnote{In particular, as every pair of preferences are trivially 0-compatible, we always have $\Pi_1 = \big\{ \{1,\ldots, N \} \big\}$.} A {\bf rearrangement} is a collection of permutations $\{\sigma_k\}_{k=1}^{\vert X \vert}$ of $\{1,\ldots, N\}$, such that $\sigma_k$ fixes each partition element of $\Pi_{k}$. In other words, a rearrangement is a choice of permutations, where each $\sigma_k$ exchanges the labels of preferences in $\succ_1,\ldots, \succ_N$, but only among those which agree on the set of $(k-1)$ best alternatives.\medskip

Let $x^k_n$ denote the $k$-th most preferred alternative of $\succ_n$. We say two sequences $\succ_1,\ldots, \succ_N$ and $\succ_1',\ldots, \succ_N'$ are {\bf rearrangement equivalent} if there exists a rearrangement $\{\sigma_k\}_{k=1}^{\vert X \vert}$ such that, for all $1 \le n \le N$:
\[
    \succ_n' \,=\, x^1_{\sigma_1^{-1}(n)}\, x^2_{(\sigma_2\, \circ \,\sigma_1)^{-1}(n)} \cdots \, x^{\vert X \vert}_{(\sigma_{\vert X \vert} \, \circ \, \cdots \, \circ \,\sigma_{1})^{-1}(n)},
\]
Rearrangement equivalence describes when one sequence of preferences can be transformed into the other by iteratively replacing compatible pairs with their conjugates. As the following example illustrates, this has a simple visual intuition: two sequences are rearrangement equivalent precisely when the strings of one can be `braided together' to form the other.

\begin{example}\label{rearrangementex}

Suppose $X = \{a,b,c,d,e,f\}$, and consider again the sequence of preferences $\succ_1, \succ_2, \succ_3$ from \autoref{sixaltexample} given by:
\[
    \begin{array}{ccc}
    \succ_1 & \succ_2 & \succ_3\\ [0.25ex] 
 \hline\hline
    \textcolor{RoyalBlue}{\textrm{a}} & \textcolor{Orange}{\textrm{b}} & \textcolor{BrickRed}{\textrm{c}}\\ 
    \textcolor{RoyalBlue}{\textrm{b}} & \textcolor{Orange}{\textrm{a}} & \textcolor{BrickRed}{\textrm{d}}\\
    \textcolor{RoyalBlue}{\textrm{c}} & \textcolor{Orange}{\textrm{e}} & \textcolor{BrickRed}{\textrm{b}}\\
    \textcolor{RoyalBlue}{\textrm{d}} & \textcolor{Orange}{\textrm{f}} & \textcolor{BrickRed}{\textrm{a}}\\
    \textcolor{RoyalBlue}{\textrm{e}} & \textcolor{Orange}{\textrm{c}} & \textcolor{BrickRed}{\textrm{f}}\\
    \textcolor{RoyalBlue}{\textrm{f}} & \textcolor{Orange}{\textrm{d}} & \textcolor{BrickRed}{\textrm{e}}\\
    \end{array}
    \quad 
    \begin{array}{ccc}
     & \Pi & \\ [0.25ex] 
 
    \tikzmark{startupp} \textcolor{white}{\ast} & \textcolor{white}{\ast} & \textcolor{white}{\ast} \tikzmark{endupp}  \\
    \ast & \ast & \ast \\
    \tikzmark{startup} \textcolor{white}{\ast} &  \textcolor{white}{\ast} \tikzmark{endup} & \ast \\
    \ast & \ast & \ast\\
    \tikzmark{startmiddle} \textcolor{white}{\ast} \tikzmark{endmiddle} & \ast  & \tikzmark{startdown} \textcolor{white}{\ast} \tikzmark{enddown}\\
    \ast & \ast & \ast \\
    \end{array}
    \quad
    \overset{\sigma}{\longrightarrow}
    \quad
    \begin{array}{ccc}
    \succ_1' & \succ_2' & \succ_3'\\ [0.25ex] 
 \hline\hline
    \textcolor{RoyalBlue}{\textrm{a}} & \textcolor{Orange}{\textrm{b}} & \textcolor{BrickRed}{\textrm{c}}\\
    \textcolor{RoyalBlue}{\textrm{b}} & \textcolor{Orange}{\textrm{a}} & \textcolor{BrickRed}{\textrm{d}}\\
    \textcolor{Orange}{\textrm{e}} & \textcolor{RoyalBlue}{\textrm{c}} & \textcolor{BrickRed}{\textrm{b}}\\
    \textcolor{Orange}{\textrm{f}} & \textcolor{RoyalBlue}{\textrm{d}} & \textcolor{BrickRed}{\textrm{a}}\\
    \textcolor{Orange}{\textrm{c}} & \textcolor{BrickRed}{\textrm{f}} & \textcolor{RoyalBlue}{\textrm{e}}\\
    \textcolor{Orange}{\textrm{d}} & \textcolor{BrickRed}{\textrm{e}} & \textcolor{RoyalBlue}{\textrm{f}}\\
    \end{array}
\]
\begin{tikzpicture}[remember picture,overlay]
\foreach \Val in {up,middle,down}
{
\draw[rounded corners,thick]
  ([shift={(-0.5\tabcolsep,-0.5ex)}]pic cs:start\Val) 
    rectangle 
  ([shift={(0.5\tabcolsep,2ex)}]pic cs:end\Val);
}
\foreach \Val in {upp}
{
\draw[rounded corners,thick]
  ([shift={(-0.5\tabcolsep,-0.5ex)}]pic cs:start\Val) 
    rectangle 
  ([shift={(0.5\tabcolsep,2ex)}]pic cs:end\Val);
}
\end{tikzpicture}

\noindent We also plot the partitions $\Pi_1,\ldots, \Pi_6$ of $\{1,2,3\}$ defined above, where $\, \ast \,$ denotes any singleton partition element. The sequence $\succ_1', \succ_2', \succ_3'$ is rearrangement equivalent to $\succ_1,\succ_2,\succ_3$, under the collection of permutations where $\sigma_1, \sigma_2, \sigma_4$, and $\sigma_6$ are the identity, and $\sigma_3$ (resp.\ $\sigma_5$) swaps the two entries of the unique non-singleton partition element of $\Pi_3$ (resp.\ $\Pi_5$). To see this, note that $\sigma_1$ and $(\sigma_2 \, \circ \, \sigma_1)$ are the identity permutation $(1,2,3)$, while $(\sigma_3 \, \circ \, \sigma _2\, \circ\, \sigma_1) = (\sigma_4 \, \circ \cdots \circ\, \sigma_1) = (2,1,3)$.\footnote{We use one-line notation for permutations, e.g., $(2,3,1)$ represents the permutation $\sigma(1) = 2$, $\sigma(2) = 3$, and $\sigma(3) = 1$. In particular, this means $\sigma^{-1}(n)$ corresponds to which component of the vector has value $n$.} Finally, $(\sigma_5 \, \circ \cdots \circ\, \sigma_1) = (\sigma_6 \, \circ \cdots \circ\, \sigma_1) = (3,1,2)$. Thus, for example, the fifth-best alternative under $\succ_2'$ is the same as for $\succ_3$, as $(\sigma_5 \, \circ \cdots \circ\, \sigma_1) = (3,1,2)$ and therefore $(\sigma_5 \, \circ \cdots \circ\, \sigma_1)^{-1} = (2,3,1)$.\hfill $\blacksquare$
\end{example}

When $\mathcal{M}$ is a support restriction, \eqref{measuretochoiceprobmap} implies the set of measures in $\mathcal{M}$ observationally equivalent to some $\mu \in \Delta(\mathcal{L})$ forms a polytope. Our next result shows that an observationally equivalent $\mu' \in \mathcal{M}$ is an extreme point of this set if and only if its support contains no pair of distinct, rearrangement-equivalent sequences.\footnote{That is, one cannot construct two such sequences by drawing (with replacement) from $\textrm{supp}(\mu')$.}  In particular, such a pair of sequences exist if and only if there are a pair of distinct rationalizations in $\mathcal{M}$ which average to $\mu'$.

\begin{theorem}\label{thm:extremepoints}
    Suppose $\mathcal{M}$ is a support restriction. Then $\mu \in \mathcal{M}$ is an extreme point of the identified set $(\mu + \mathcal{R}) \cap \mathcal{M}$ if and only if $\textrm{supp}(\mu)$ contains no pair of distinct, rearrangement-equivalent sequences.
\end{theorem}

Whereas \autoref{ryserspacetheorem} characterizes the identified sets of the random utility model in terms of Ryser swaps, in the case of support restrictions, \autoref{thm:extremepoints} provides a complementary description via their extreme points.\footnote{When $\mathcal{M}$ is a support restriction, \autoref{ryserspacetheorem} and \autoref{thm:extremepoints} provide dual characterizations of the identified polytopes, in terms of their faces and extreme points, respectively. We note, however, that \autoref{ryserspacetheorem} remains valid without assumptions on $\mathcal{M}$.}  These extremal distributions are of natural interest, as they constitute the primitive, irreducible rationalizations from which any others may be constructed. In practice, they often correspond to the simplest, or most parsimonious, population models which explain the data. For example, single-crossing (\citealt{apesteguia2017single}), progressive (\citealt{filiz2023progressive}), and swap-progressive (\autoref{sec:Ordered}) random utility representations are always extreme points of their identified sets.\footnote{See \autoref{prop:swapProgisExt}. However, there exist identified sets whose extreme points do not all arise in this manner; see \autoref{prop:ExtnotOrdered}.}\medskip

\autoref{thm:extremepoints} also resolves an open question in the theory of rational \emph{joint} choice behavior.\footnote{E.g.\ \cite{chambers2024correlated}.} \cite{kashaev2024entangled} show that, while the revealed preference implications of such models are generally intractable, they can be simply axiomatized under any restriction to  sets of preferences with linearly independent choice functions.\footnote{Here, we are identifying choice functions with their representation as $0\!-\!1$ vectors in $\mathcal{P}$.} To date, no characterization of these sets has been known. However, a key ingredient in our proof is showing they are precisely the sets which contain no pair of distinct, rearrangement-equivalent sequences. Thus our results also describe when testing for stochastically rational joint behavior is easy.\medskip

In addition, this allows us to completely describe identifying support restrictions.

\begin{corollary}\label{thm:braids}
    A support restriction defined by $\mathcal{S} \subseteq \mathcal{L}$ is identifying if and only if $\mathcal{S}$ contains no pair of distinct, rearrangement-equivalent sequences.
\end{corollary}

To illustrate, consider again the support restriction $\mathcal{S}$ consisting of the four preferences in \autoref{introex}.  As noted before, the only admissible permutation in any rearrangement that can generate new preferences is $\sigma_2$, thus two sequences are rearrangement equivalent if and only if one can be obtained from the other by shuffling terminal segments $s^2_\downarrow$. As a consequence, while $\mathcal{S}$ is not identifying, any proper subset is. In particular, a distribution is an extreme point of its identified set if and only if it does not have full support; see \autoref{fig:ProdAndSimp}.\medskip

Similarly, \autoref{thm:braids} also implies that the restriction to monotone expected utility preferences in \autoref{ex:MonotoneEU} is identifying. As in \autoref{introex}, it once again suffices to consider only rearrangements with $\sigma_2$ the sole non-trivial permutation. As a consequence, proving identification reduces to verifying one cannot construct two distinct sequences of preferences from $\mathcal{S}$ with identical initial $2$-segments, but whose corresponding sequences of terminal $2$-segments are distinct. But such sequences cannot exist, essentially due to the fact each preference in $\mathcal{S}$ which ranks $b \succ a$ does not also rank $c \succ d$. Thus $\mathcal{S}$ is identifying.

\section{Order-Based Restrictions}\label{sec:Ordered}

In many economic environments, the underlying set of alternatives carries a natural order structure.  For example, monetary lotteries may be ordered by various measures of risk or stochastic dominance, vehicles by gas mileage, or government policies by level of public good provision.  In such contexts, behavioral types often possess an analogous ordering, e.g.\ by risk aversion, environmental conscientiousness, or civic-mindedness.\medskip

Throughout this section, we will assume that $\trianglerighteq$ is a linear order on $X$. Following \cite{apesteguia2017single}, a distribution $\mu \in \Delta(\mathcal{L})$ is a {\bf single-crossing} representation if $\textrm{supp}(\mu)$ can be ordered $\succ_1,\ldots, \succ_N$ such that if (i) $x \trianglerighteq y$ and (ii) $x \succ_j y$, then for all $i > j$, $x \succ_i y$ as well.  When a single-crossing representation exists, \cite{apesteguia2017single} show is necessarily unique; however there exist choice frequencies which are consistent with the random utility model but incompatible with any single-crossing representation.\medskip

An alternative approach to identification due to \cite{filiz2023progressive} considers a broader class of primitives.  Let $\mathcal{C}$ denote the set of all choice functions, including those not rationalizable by any strict preference. A distribution $\mu \in \Delta(\mathcal{C})$ is said to be a {\bf progressive} random choice representation if $\textrm{supp}(\mu)$ can be ordered $c_1,\ldots, c_N$ such that $i > j$ implies $c_i(A) \trianglerighteq c_j(A)$ for all menus $\varnothing \subsetneq A \subseteq X$. In contrast to single-crossing representations, progressive random choice rules exist and are unique for any choice shares. However, they may not be supported on the rational choice functions, and hence fail to be a random utility representation, even for data compatible with the random utility model.\medskip

Our notion of compatibility provides a novel, order-based identifying restriction which preserves the best features of both the single-crossing and progressive frameworks.  We say that a random utility representation $\mu \in \Delta(\mathcal{L})$ is {\bf swap-progressive} if $\textrm{supp}(\mu)$ can be ordered $\succ_1,\ldots, \succ_N$ such that, whenever $\succ_i$ and $\succ_j$ are $k$-compatible, $0 \le k \le \vert X \vert-1$, then $i > j$ implies $x^{k+1}_i \trianglerighteq x^{k+1}_j$.\footnote{Recall that $x^k_n$ denotes the $k$-th most preferred alternative of $\succ_n$, and note $x^{k+1}_i \trianglerighteq x^{k+1}_j$ explicitly allows for the equality of $x^{k+1}_i$ and $x^{k+1}_j$.}\medskip

Our next result shows swap-progressive representations are unique, and exist if and only if the data admit a random utility rationalization. Thus the empirical content of swap-progressivity coincides with that of the baseline random utility model.  Moreover, when either single-crossing or progressive random utility representations exist, both do, and both coincide with the unique swap-progressive representation.\footnote{The first part of this claim is Remark 1 in \cite{filiz2023progressive}.}

\begin{theorem}\label{thm:OrderedRUM}
Let $\trianglerighteq$ be a linear order on $X$. Then:
        \begin{itemize}
            \item[(i)] A swap-progressive rationalization exists if and only if the data are compatible with the random utility model. Moreover, if such a representation exists, it is unique.
            \item[(ii)] For any order $\trianglerighteq$, every single-crossing representation is swap-progressive. Conversely, a swap-progressive representation is single-crossing if and only if it is progressive.
        \end{itemize}
\end{theorem}

In light of \autoref{thm:OrderedRUM}, we interpret swap-progressivity as a principled refinement of both notions: it agrees with the random utility predictions of both single-crossing and progressivity whenever these exist, but generates unique predictions even when these other frameworks cannot. This provides a practical, but general, method for selecting rationalizations, that systematically incorporates the order structure of the environment.\medskip

In particular, in many applications it is valuable to obtain rationalizations whose support is `1-dimensional,' or able to be linearly ordered according to some economic criterion.\footnote{See, e.g., \cite{chiappori2019aggregate, barseghyan2021discrete}. For further examples and a discussion of these ideas, see \cite{filiz2023progressive}.} As the following example illustrates, this is always true of swap-progressive rationalizations.\footnote{For a more complex example of the swap-progressive representation of data which do not admit a single-crossing representation, see \hyperlink{OrderedFish}{Online Appendix D.2}.}

\begin{example}\label{ex:OrderedFishConstruction}
    Consider again the setting of \autoref{introex}, and suppose policies are ranked according to their fiscal cost $a \trianglerighteq b \trianglerighteq d \trianglerighteq c$.  \autoref{thm:OrderedRUM} guarantees the policy maker is able to recover a unique swap-progressive rationalization from straw-poll data, even though infinitely many rationalizations exist.\medskip

    The power of swap progressivity is that it severely limits the support of a rationalization. For example, no swap-progressive distribution can place positive weight on both $\succ_1 \, = abcd$ and $\succ_2 \, = badc$. As both are trivially $0$-compatible, $a \trianglerighteq b$ implies $\succ_1$ must be ranked higher within the support than $\succ_2$. However, as both are also 2-compatible, $d \trianglerighteq c$ implies $\succ_2$ must also be ranked above $\succ_1$, an impossibility.\medskip

    In this example, this observation already suffices to pin down the distribution: the $\trianglerighteq$-swap progressive rationalization is precisely $\bar{\mu} = (0,0,\frac{5}{8}, \frac{3}{8})$. In particular, $\bar{\mu}$ is supported on $bacd$ and $abdc$. These correspond to the two preferences, among those held in \emph{some} rationalization, which (resp.) agree least and most with $\trianglerighteq$. As a consequence, we can regard the support of $\bar{\mu}$ as being ordered by fiscal conservativeness.\hfill $\blacksquare$
\end{example}

\section{Parametric Restrictions}\label{sec:Parm}

In practical applications, restrictions on random utility, and more general models of stochastic choice, are often specified parametrically.  If $\Theta \subseteq \mathbb{R}^d$ denotes a set of parameters, we refer to any one-to-one and onto map $F: \Theta \to \mathcal{M} \subseteq \Delta(\mathcal{L})$ as a {\bf parametric restriction}.\footnote{This is distinct from the case in which a support constraint is defined by the restriction of some parametric family of \emph{utilities} to $X$.  There, parameters correspond to the preferences of the model, and hence rationalizations to distributions over parameters.  Here, we consider the case in which the allowable population distributions \emph{themselves} belong to some parametric family $\mathcal{M} = F(\Theta) \subseteq \Delta(\mathcal{L})$.} By \autoref{ryserspacetheorem}, a parametric restriction is identifying if and only if:
\begin{equation*}\label{parametricrumid}
    F(\theta') \in F(\theta) + \mathcal{R} \quad \implies \quad \theta = \theta',
\end{equation*}
i.e.\ if $\mathcal{M} = \textrm{range}(F)$ intersects each translate of the Ryser subspace at most once. In principle, this is a complete characterization of identification.  However, in practice, it is often unclear how to verify this from only knowledge of $F$. Moreover, many times restrictions are defined by specifying choice probabilities directly in terms of parameters, leaving $F$ only implicit. We now consider a simple, practical test for identification in such cases.\medskip

We will consider only \emph{smooth} parametric models. However, as it will essentially be costless to do so, we will allow for arbitrary models of parametric stochastic choice, rather than solely random utility. Recall $\mathcal{P}$ denotes the set of all random choice rules. A {\bf parametric stochastic choice model} is a smooth mapping $\varphi: \Theta \to \mathcal{P}$.\footnote{Parametric random utility corresponds to the case where $\varphi = \Phi \circ F$, for some one-to-one $F: \Theta \to \Delta(\mathcal{L})$.} In general, it is intractable to directly analyze the full vector of functions $\varphi$ in any realistically-sized setting.\footnote{For example, when $\vert X \vert = 6$, a parametric stochastic choice model is a vector of 129 functions.} To determine invertibility in practice, one typically restricts to subsystems with equal numbers of parameters and choice probabilities.\medskip

We say $\bar{\varphi}: \Theta \to \mathbb{R}^d$ defines a {\bf smooth submodel} if it is obtained by selecting any $d$ components of the vector-valued map $G \circ \varphi$, where $G$ is any smooth and smoothly invertible transformation of choice probabilities. As there are typically many fewer parameters than choice probabilities, this entails a meaningful reduction. While our definition allows for general transformations, in practice, $G$ often corresponds to simply normalizing by certain choice probabilities, see \autoref{parametricluce} and \autoref{chriscex}.\medskip 

A smooth submodel $\bar{\varphi}$ is {\bf parametrically identified} if it admits a smooth inverse $\bar{\varphi}(\Theta) \to \Theta$. Parametric identification of a submodel $\bar{\varphi}$  implies the same of the full, unrestricted model $\varphi$. While in full generality the converse need not be true, in practice it is typically straightforward to use an unidentified $\bar{\varphi}$ to show the full model $\varphi$ fails to be parametrically identified.\footnote{See, e.g., \autoref{chriscex} and \hyperlink{chriscexworked}{Online Appendix D.1}.}\medskip

Our next result shows a smooth submodel is parametrically identified precisely when it is so locally, and it satisfies a mild, boundary-preservation condition. This reduces the difficult, global problem of verifying identification of a submodel to a simple analytic test.

\begin{theorem}\label{thm:paramthm}
    Let $\Theta \subseteq \mathbb{R}^d$ be a contractible, open set of parameters.\footnote{A set $\Theta$ is contractible if there exist a continuous map $H: \Theta \times [0,1] \to \Theta$ such that for each $\theta \in \Theta$, $H(\theta, 0) = \theta$, and $H(\theta,1) = \theta_0$ for some fixed vector $\theta_0$. This includes all convex (or star-shaped) sets.} A smooth map $\bar{\varphi}: \Theta \to \mathbb{R}^d$ is parametrically identified if and only if:
    \begin{itemize}
        \item[(i)] The Jacobian matrix $\mathcal{J}_{\bar{\varphi}}$ is everywhere invertible; and
        \item[(ii)] For any sequence $(\theta_n)_{n=1}^\infty \subset \Theta$ such that each compact $K \subset \Theta$ contains at most finitely many terms, the same is true of $\big(\bar{\varphi}(\theta_n)\big)_{n=1}^\infty$ for any compact $K'\subset \bar{\varphi}(\Theta)$.
    \end{itemize}
\end{theorem}

Relative to existing global inverse results, \autoref{thm:paramthm} has several key advantages. First, unlike approaches relying on monotonicity, which typically yield sufficient but not necessary conditions, it provides a complete characterization of parametric identification. For example, when studying the invertibility of a smooth demand system, \autoref{thm:paramthm} remains valid even when goods may be complements, and is applicable in settings when the law of demand may fail.\footnote{In particular, \autoref{thm:paramthm} remains applicable even in economic environments not amenable to classic demand inversion results, e.g.\ \citealt{gale1965jacobian, berry2013connected, allen2022injectivity}.}\medskip

Second, unlike many existing results, \autoref{thm:paramthm} allows for both the domain and range of $\bar{\varphi}$ to be distinct, proper subsets of $\mathbb{R}^d$, and requires structural assumptions only on the \emph{domain}, rather than range of $\bar{\varphi}$.\footnote{C.f.\ \citealt{palais1959natural, mckenzie1967theorem, gordon1972diffeomorphisms, krantz2002implicit, komunjer2012global}.} Given that the range of a general vector-valued function can be difficult to determine, this property makes \autoref{thm:paramthm} particularly valuable in practice.

\begin{example}\label{parametricluce}
    Let $X = \{x_0,\ldots, x_K\}$. A stochastic choice rule belongs to the {\bf Luce model} if there exist weights $w_0 \equiv 1$ and $w_k > 0$ for all $k > 0$, such that, for any $x \in A \subseteq X$:
    \[
        \varphi(w_1,\ldots, w_K)_{(x_i,A)} = \rho(x_i,A) = \frac{w_i}{\sum_{j \in A} w_j}.
    \]
    Thus $\Theta = \mathbb{R}^{K}_{++}$, with each $\theta$ corresponding to the vector of normalized weights $(w_1,\ldots, w_K)$.  Consider the map:
    \[
        \bar{\varphi}(w_1, \ldots, w_K) = \begin{bmatrix}
            \frac{\rho(x_1,X)}{\rho(x_0,X)}\\
            \vdots \\
            \frac{\rho(x_K,X)}{\rho(x_0,X)}
        \end{bmatrix} = \begin{bmatrix}
            w_1\\
            \vdots \\
            w_K
        \end{bmatrix},
    \]
    noting this defines a smooth submodel.\footnote{Here, the map $G$ takes each component of $\varphi$ and normalizes it by $\varphi_{\rho(x_0, X)}$; $\bar{\varphi}$ is then obtained by selecting the $d = K$ components corresponding to normalized choice probabilities for $x_1, \ldots, x_K$ from $X$.} Since (i) the Jacobian of $\bar{\varphi}$ is everywhere the identity, and (ii) some component of $\bar{\varphi}(\theta_n)$ trivially tends to zero or infinity whenever the same is true of $\theta_n$, it follows from \autoref{thm:paramthm} that the submodel $\bar{\varphi}$, and hence the full Luce model, is parametrically identified.\hfill $\blacksquare$
\end{example}

Our results can also be used to detect identification failures in more complex settings. To illustrate, we establish a novel result, that the dynamic, logit habit-formation model of \cite{TuransickIntertemporal2025} is not parametrically identified from unconditional choice data.

\begin{example}\label{chriscex}
    Suppose $X = \{x_0,\ldots, x_N\}$, where $x_0$ denotes an outside option which is always available.  A menu is any non-empty subset of $X$ which contains $x_0$.  A stochastic choice rule represents the unconditional choice probabilities of the habit-formation logit model of \cite{TuransickIntertemporal2025} if it takes the form:
    \[
        \rho(x_i,A) = \frac{v_i \bigg(\sum_{j \in A \setminus i} v_j  + v_i c_i\bigg)}{\sum_{j \in A} v_j\bigg(\sum_{k \in A \setminus j} v_k + v_j c_j\bigg)},
    \]
    for any menu $A$, where $v_0 = c_0 \equiv 1$, and $v_i, c_i > 1$ for all $ i=1 ,\ldots, N$. Here, the parameters $v_i$ capture the desirability of consuming each alternative, while $c_i$ speaks to how addicting, or habit-forming, consumption of $x_i$ is in a given period.\medskip

    Define a submodel as follows. For $i = 1,\ldots, N$, let:
        \[
            \bar{\varphi}^i(v_1,c_1,\ldots, v_N, c_N) = \frac{\rho\big(x_0, \{x_0, x_i\}\big)}{\rho\big(x_i, \{x_0, x_i\}\big)} = \frac{1+v_i}{v_i(1+v_i c_i)},
        \]
        and for $i = N\!+\!1, \ldots, 2N$:
        \[
            \bar{\varphi}^{N+i}(v_1,c_1,\ldots, v_N, c_N) = \frac{\rho\big(x_0, \{x_0, x_i, x_{i+1}\}\big)}{\rho\big(x_i, \{x_0, x_i, x_{i+1}\}\big)} = \frac{1+v_i + v_{i+1}}{v_i + v_iv_{i+1} +c_i v_i^2},
        \]
    where $v_{N+1} \equiv v_1$, and $c_{N+1} \equiv c_1$.\medskip
    
    We now show this submodel fails to satisfy condition (ii) of \autoref{thm:paramthm}, and hence fails to be parametrically identified. To do so, we first define the function:
        \[
            c(v_i) = \frac{1+v_i}{10v_i^2} - \frac{1}{v_i}.
        \]
    Note that, for any $i = 1, \ldots, N$, we have $\bar{\varphi}^i(\ldots, v_i, c(v_i), \ldots) = \frac{1}{10}$. Consider now the restriction of $\bar{\varphi}$ to the curve $\big(v, c(v), \ldots, v, c(v) \big)$. The last $N$ components of $\bar{\varphi}$ are identical, and equal to:
        \[
            \bar{\varphi}^{N+i}\big(v,c(v), \ldots, v, c(v)\big) = \frac{1+2v}{v^2 + 10(1+v)}.
        \]
    As $v \to 1$ along this curve, $\bar{\varphi}^{N+i} \to\frac{1}{7}$. But this is precisely the value attained by $\bar{\varphi}^{N+i}$ at $v=3$.  Since the first $N$ components of $\bar{\varphi}$ are constant in $v$, we have shown that for any sequence of parameters of the form $\big(v_n, c(v_n), \ldots, v_n, c(v_n)\big)$ with $v_n \to 1$, we obtain a violation of condition (ii) of \autoref{thm:paramthm}, and hence $\bar{\varphi}$ is not parametrically identified. It is then straightforward to use the insight from this submodel to construct a distinct pair of parameter vectors that yield identical choice probabilities in the full model.\footnote{See \hyperlink{chriscexworked}{Online Appendix D.1} for an explicit example.} \hfill $\blacksquare$
\end{example}

\section{Beyond Random Utility}\label{sec:Extensions}

Thus far, we have considered primarily the random utility model. However, our core notion of Ryser swap characterizes the identified sets of many other stochastic choice models. In each of these additional settings, exact analogs of our main results obtain. We illustrate several such extensions below.

\subsubsection*{Random Choice}

The simplest extension is to the unrestricted class of random choice models, possibly subject to limited observability.  Here, $X$ is a finite set of alternatives, and $\Sigma \subseteq 2^X \setminus \{\varnothing\}$ an arbitrary collection of menus on which choice frequencies are observed.  Let $\mathcal{C}$ denote the collection of all choice functions on $\Sigma$.\footnote{That is, $\mathcal{C}_\Sigma = \prod_{A \in \Sigma} A$.} A random choice model is a probability measure $\mu \in \Delta(\mathcal{C}_\Sigma)$ such that, for all $a \in A \in \Sigma$:
\[
    \rho(a,A) = \sum_{c \in \mathcal{C}} \mu(c)\, \mathbbm{1}_{\{a \, = \, c(A)\}}.
\]
Here, a Ryser swap is a signed measure $R \in \mathbb{R}^{\mathcal{C}_\Sigma}$ of the form:
\[
    R = \mathbbm{1}_{\{c'_1,c'_2\}}  - \mathbbm{1}_{\{c_1,c_2\}},
\]
where $c'_1,c'_2 \in \mathcal{C}_\Sigma$ are the choice functions obtained by swapping the choices made by $c_1$ and $c_2$ on a single menu $A \in \Sigma$. Similarly, we define the Ryser subspace for the random choice model $\mathcal{R}^{rc} \subset \mathbb{R}^{\mathcal{C}_\Sigma}$ as the linear span of all such vectors. Once again, the subspace $\mathcal{R}^{rc}$ characterizes identifying restrictions.

\begin{theorem}\label{thm:SCgeneral}
    Fix $\Sigma \subseteq 2^X\setminus \{\varnothing\}$. Two distributions $\mu, \mu' \in \Delta(\mathcal{C}_\Sigma)$ generate identical choice probabilities on every menu in $\Sigma$ if and only if $\mu - \mu' \in \mathcal{R}^{rc}$. In particular, for any $\mathcal{M} \subseteq \Delta(\mathcal{C}_\Sigma)$ the set of distributions generating identical choice probabilities to $\mu$ on every menu in $\Sigma$ is precisely $\big(\mu + \mathcal{R}^{rc}\big) \cap \mathcal{M}$.
\end{theorem} 

Exact analogues of \autoref{thm:extremepoints} and \autoref{thm:OrderedRUM} also obtain in this setting.\footnote{In particular, the analogue of our notion of swap-progressivity coincides precisely with the progressivity condition of \citet{filiz2023progressive}; the analogue of our \autoref{thm:OrderedRUM} in this setting corresponds their Theorem 1.} Notably, since $\mathcal{C}_\Sigma$ contains the rational choice functions on $\Sigma$, the random utility model can be viewed as a support restriction within this more general framework. Thus \autoref{thm:SCgeneral}, and the random choice analogue of \autoref{thm:extremepoints}, allow our results to extend to the case of incompletely observed data, i.e.\ when $\Sigma \subsetneq 2^X \setminus\{\varnothing\}$.

\subsubsection*{Frame Dependent Choice}
Our results also extend naturally to models of boundedly rational behavior. We illustrate this on a simple version of the frame-dependent random utility model of \citet{Cheung2024FRUM}.\footnote{We extend this analysis to the full model of \cite{Cheung2024FRUM} in     \hyperlink{frumgraph}{Online Appendix B.2}.} Here, individuals face a fixed, finite choice set $X= \{x_1,\ldots, x_K\}$, but receive varying suggestions from a black-box recommendation system that highlights some subset $A \subseteq X$, rendering them more desirable.\medskip

Here, the empiricist observes \emph{frame-dependent} stochastic choice data
$\rho: X \times 2^X \to [0,1]$, where:
\[
        \sum_{x \in X } \rho(x,A) = 1,
\]
for each $A \subseteq X$, interpreted as the observed probability of the the agent choosing $x$ from $X$ when faced with the recommendation $A$.\medskip

Define $X^* = \{x_1^N, \ldots, x_K^N, x_1^F, \ldots, x^F_K\}$, where we regard $x_i$ when framed ($x_i^F$) as distinct from $x_i$ when it is not ($x_i^N$). Given a recommendation set $A \subseteq X$, the subject may be regarded as choosing from the virtual menu:
\[
    A^* =  \{x^N : x \in X \setminus A\} \, \cup \, \{x^F : x \in A\} \subset X^*.
\]
Since framing makes any alternative strictly more preferred, no alternative below the most desirable $x^N$ can ever affect observed choice. Thus, let $\mathcal{L}^*$ denote the set of \emph{truncated} strict preferences on $X^*$, consisting only of the ranking up to, and including, the best, non-framed alternative, and which satisfy $x^F \succ^* x^N$ for all $x \in X$.\footnote{For example, an (unrestricted) preference $x_1^F \succ x_2^F \succ x_1^N  \succ x_2^N$ would correspond to the truncated preference $x_1^F \succ^* x_2^F \succ^* x_1^N$.}\medskip

The data are compatible with a frame-dependent random utility representation $\mu \in \Delta(\mathcal{L}^*)$ if:
\[
\rho(x,A) = \sum_{\succ^* \in \mathcal{L}^*} \mu(\succ^*) \mathbbm{1}_{\{\hat{x} \textrm{ is $\succ^*$-maximal in $A^*$}\}},
\]
where $\hat{x} = x^N$ if $x \not \in A$, and $x^F$ if $x \in A$. Given two truncated preferences in $\succ_1^*, \succ_2^* \in \mathcal{L}^*$ we once again say they are compatible if they agree upon the set of $k$-best alternatives, though not necessarily their ranking.  A Ryser swap for the frame-dependent model is simply a signed measure in $\mathbb{R}^{\mathcal{L}^*}$ which places unit mass on a $k$-compatible pair $\succ_1^*$ and $\succ_2^*$ and mass negative one on the pair of truncated preferences obtained by swapping their $k$-initial segments.\footnote{Note this is well-defined, as these swaps necessarily belong to $\mathcal{L}^*$.}  Let $\mathcal{R}^{fd} \subset \mathbb{R}^{\mathcal{L}^*}$ denote the linear span of these vectors.

\begin{theorem}\label{thm:FRUMcor}
    Two distributions $\mu, \mu' \in \Delta(\mathcal{L}^*)$ are observationally equivalent if and only if $\mu - \mu' \in \mathcal{R}^{fd}$. In particular, for any $\mathcal{M} \subseteq \Delta(\mathcal{L}^*)$, its identified set is precisely $ \big(\mu + \mathcal{R}^{fd}\big) \cap \mathcal{M}$.
\end{theorem}

As before, analogues of \autoref{thm:extremepoints} and \autoref{thm:OrderedRUM} also obtain in this setting. In particular, as was true for the random utility model, a swap-progressive \emph{frame-dependent} representation exists whenever the data are rationalizable, and will itself always belong to $\Delta(\mathcal{L}^*)$.

\subsubsection*{Dynamic Discrete Choice}\label{sec:DDC}

Our methodology is also applicable to the problem of identifying distributions over dynamic choice functions from Markovian conditional choice probabilities. Let $X$ denote a finite set of alternatives, and $T$ a finite time horizon.  In period $t=1$, the analyst observes unconditional choice probabilities $\rho_1(x)$. For each period $t= 2,\ldots ,T$, the analyst observes probabilities of time $t$ choice, conditional the prior period's choice,  $\rho_t(y \, \vert \, x)$. Together, we refer to $\{\rho_1,\dots,\rho_T\}$ as a system of conditional choice probabilities.\medskip

An analyst seeks to decompose a system of conditional choice probabilities into a measure over dynamic choice functions in $\mathcal{C}^{dc} = X^T$.  A system of conditional choice probabilities is compatible with a dynamic discrete choice representation $\mu \in \Delta(\mathcal{C}^{dc})$ if:
\[
    \rho_1(x) = \sum_{c \in \mathcal{C}^{dc}} \mu(c) \mathbbm{1}_{\{x = c(1)\}}
\]
and
\[
    \rho_t(y \, \vert \, x) = \frac{\sum_{c \in \mathcal{C}^{dc}} \mu(c) \mathbbm{1}_{\{x \, =\, c(t-1), \, y \,=\, c(t) \}}}{\sum_{c \in \mathcal{C}^{dc}} \mu(c) \mathbbm{1}_{\{x\, =\, c(t-1)\}}}
\]
for all $ t = 2, \ldots, T$ and $x,y \in X$.\medskip

We say two choice functions $c_1, c_2 \in \mathcal{C}^{dc}$ are $t$-compatible if $c_1$ and $c_2$ make the same choice at time $t\!-\!1$.  A Ryser swap is a signed measure $R \in \mathbb{R}^{\mathcal{C}^{dc}}$ of the form $\mathbbm{1}_{\{c_1', c_2'\}} - \mathbbm{1}_{\{c_1, c_2\}}$, where $c'_1$ and $c_2'$ are the choice functions obtained by swapping the first $t\!-\!1$ choices of some $t$-compatible pair of choice functions $c_1, c_2 \in \mathcal{C}^{dc}$. Let $\mathcal{R}^{dc} \subset \mathbb{R}^{\mathcal{C}^{dc}}$ denote the linear span of these measures.

\begin{theorem}\label{thm:DDCObsEq}
    Two distributions $\mu, \mu' \in \Delta(\mathcal{C}^{dc})$ are observationally equivalent if and only if $\mu - \mu' \in \mathcal{R}^{dc}$. In particular, for any $\mathcal{M} \subseteq \Delta(\mathcal{C}^{dc})$, its identified set is precisely $ \big(\mu + \mathcal{R}^{dc}\big) \cap \mathcal{M}$.
\end{theorem}

Finally, as in each previous setting, analogues of \autoref{thm:extremepoints} and \autoref{thm:OrderedRUM} also obtain. In particular, a swap-progressive rationalization exists and is unique for any rationalizable data, and will always be supported on $\mathcal{C}^{dc}$.\medskip

While we consider here dynamic discrete choice, our results also extend to more complex, finite dynamic programming models, including those involving state-dependent utility and menu variation. This allows for further application to models such as \cite{rust1987}.

\section{Related Literature}\label{sec:RelatedLit}

The first explicit example of identification failure in the random utility model is due to \citet{fishburn1998stochastic}. His counterexample forms the basis of our notion of a (non-trivial) compatible pair. \citet{turansick2022identification} characterizes data sets which possess a unique random utility rationalization. We obtain this  characterization as a straightforward consequence of the proof of our \autoref{ryserspacetheorem}.\footnote{See also \cite{doignon2023adjacencies}.}\medskip

A number of papers consider restrictions on the random utility model for purposes of obtaining identification. \cite{manzini2018dual} show that in a two-state model, distributions over preferences can generically be recovered. \citet{suleymanov2024branching} recovers identification by requiring certain independence properties of the rationalization; \citet{honda2021random} considers a random cravings model where identification obtains under monotonicity restrictions. By directly characterizing the identified sets of the general random utility model, we instead provide necessary and sufficient conditions for any restriction to be identifying.\medskip 

When alternatives are ordered, a well-known class of restrictions rely on various single-crossing type conditions on the support of rationalizations to achieve uniqueness. This idea is originally due to \cite{apesteguia2017single}; \citet{filiz2023progressive} considered an extension to general stochastic choice models that guarantees existence while preserving uniqueness, at the cost of the unique rationalization possibly failing to belong to the random utility model, even when the data is rationalizable by such a distribution.\footnote{\cite{yildiz2023foundations} studies the structure of such ordered models and shows that a unique ordered representations is possible only when the types in the model form a lattice. On the other hand, \citet{petri2022random} considers an extension of these ordered models to allow for multi-valued choice.}\medskip

Our notion of swap-progressivity builds on these contributions. Applied to the random utility model, it generalizes single-crossing and provides unique predictions, even for data which fail the centrality axiom of \cite{apesteguia2017single}.\footnote{In particular, in this setting, swap progressivity has the same explanatory power as the baseline random utility model, while also being identified.} In the context of general random choice, the analogue of our \autoref{thm:OrderedRUM} is precisely Theorem 1 of \cite{filiz2023progressive}. Moreover, swap-progressive representations exist, are unique, and belong to the models of interest for each of the environments considered in \autoref{sec:Extensions}.\medskip

In practice, it is common to specify random utility model parametrically, e.g.\ \citet{luce1959individual}.\footnote{This model is observationally equivalent to the logit model of \citet{mcfadden1974conditional}. See also \cite{sandomirskiy2023decomposable} for a recent characterization.} More recently, \citet{chambers2024weighted} considers an extension of the Luce model including both salience and utility in its parameterization which preserves identification. In general, Luce type structure allows for simple identification arguments. However, in more complex environments, \autoref{thm:paramthm} provides a practical tool for establishing similar results.\medskip

While our paper focuses on the case of finite environments, there are numerous positive identification results when domains are infinite. \citet{gul2006random} shows that the random expected utility model is identified, as is the random quasilinear utility model \citep{williams1977formation, daly1979identifying, yang2023random}.\medskip

Methodologically, our \autoref{ryserspacetheorem} builds on work in the mathematical discipline of tomography (\citealt{ryser1957combinatorial, fishburn1991sets, Kong1999}).\footnote{\cite{heckman2018unordered} use Ryser's theorem to show identification in a different context. A different strand of literature uses tensor-decomposition techniques to study identification; see \citet{dardanoni2023mixture} and \cite{Kops2026behavioraltypes}.} We show that both classic tomographic applications, as well as related questions about certain, restricted sets of measures with prescribed marginals, can be be fruitfully encoded as flow problems on directed graphs. In this general setting, we provide a unified notion of Ryser swap, which specializes to the forms seen in \autoref{sec:GeometryResults} and \autoref{sec:Extensions}, and show it characterizes the abstract analogue of our identified sets.\medskip

Similarly, to the extent of our knowledge, \autoref{thm:paramthm} does not appear in the economic or mathematical literature. Relative to classical results of Hadamard (\citealt{gordon1972diffeomorphisms, krantz2002implicit}), we show that when the domain of a non-linear system is convex (or, more generally, contractible) and the system contains an equal number of equations and unknowns, the classical, difficult-to-verify conditions on the system's range which characterize invertibility are always satisfied. Our proof relies on a fixed point theorem due to \cite{smith1934theorem} which we believe has not found prior economic application.

\section{Conclusion}\label{sec:Conc}

When analysts observe only aggregate choice frequency data, many work-horse models of stochastic choice are unidentified. In such instances, explicit descriptions of the identified sets are of direct practical importance, as they allow analysts not only to understand the range of possible aggregate effects of a policy intervention, but also its distributional implications. Despite this, little has been known about the structure of these sets, even for many widely used models.\medskip

We characterize the identified sets of a wide range of stochastic choice models. Our key insight is that any transformation of a distribution over primitives that preserves all choice probabilities can be represented as a finite sequence of simple, mass-swapping operations that we term \emph{Ryser swaps}. Using this, we give explicit characterizations of both the defining inequalities and extreme points of these models' identified sets. In each of these settings, our results provide tight bounds on the population cross-sections compatible with observed choice frequencies. This yields valuable insights, e.g., for assessing the effects of policies which may differentially affect various subsets of a population.\medskip

We also provide specialized characterizations of identification for a number of common, structured classes of restriction, such as constraints on the support of the rationalizing distribution, or when choice frequencies are assumed to vary smoothly in some vector of parameters. These provide novel tools for practitioners to employ when evaluating the identification properties of various restrictions going forward.\medskip

Our results suggest several fruitful directions for further inquiry. One would be to understand how our methodological techniques generalize to the case of continuous consumption spaces.  While we have emphasized that our results apply broadly across various economic domains of interest, another interesting question would be to determine precisely which models admit a notion of Ryser swap which characterizes their identified sets. We leave these questions for future research.

\bibliographystyle{ecta}
\bibliography{linind}

\appendix

\section{Graphical Results}

\subsection{Definitions \& Preliminaries}\label{app:prelims}

A {\bf graph} is a pair $(\mathcal{N}, \mathcal{E})$ where $\mathcal{N}$ is a finite set of {\bf nodes}, and $\mathcal{E}$ a finite collection of ordered pairs in $\mathcal{N} \times \mathcal{N}$, where we explicitly allow for repetition.\footnote{That is, $\mathcal{E}$ is a multi-set, hence our definition of graph explicitly allows for \emph{multigraphs}.}  We interpret an ordered pair $(x,y) \in \mathcal{E}$ as directed from $x$ (the `tail') to $y$ (the `head').  We will sometimes denote the ordered pair $(x,y) \in \mathcal{E}$ by $x \to y$.\medskip

For a given node $n \in \mathcal{N}$, we write $\mathcal{E}^\downarrow_n$ (resp. $\mathcal{E}^n_{\,\downarrow}$) for the sub-collection of all edges in $\mathcal{E}$ whose head (resp. tail) is $n$. A graph is said to be {\bf acyclic} if there is no subset of edges with the property that any node appearing in some edge in the set does so precisely once as the head of some edge, and once as a tail. A {\bf directed acyclic graph} (DAG) is a tuple $(\mathcal{N}, \mathcal{E}, s, t)$, where $(\mathcal{N}, \mathcal{E})$ is an acyclic graph, and $s$ and $t$ denote specialized nodes in $\mathcal{N}$, respectively the {\bf source} and {\bf sink}, such that $\mathcal{E}_s^\downarrow = \mathcal{E}^t_{\,\downarrow} = \varnothing$. We will assume that $s$ and $t$ are the unique nodes with these properties.\vspace{-2mm}

\subsubsection*{Flows \& Path Decompositions}
Given a directed acyclic graph,  a {\bf flow} is a function $f : \mathcal{E} \to \mathbb{R}_+$ such that (i) at every non-source/sink node $n$, total in-flow equals total out-flow, and (ii) the total out-flow from the source is normalized to unity:
\begin{equation}\label{flowdef}
    \begin{aligned}
         \textrm{(i)} \;  \; \sum_{e \in \mathcal{E}^\downarrow_n} f(e) = \sum_{e \in \mathcal{E}^n_{\, \downarrow}} f(e) & \quad \textrm{ and } \quad  \textrm{(ii)} \; \; \sum_{e \in \mathcal{E}^s_{\downarrow}} f(e) = 1.
    \end{aligned}  
\end{equation}
A function $f : \mathcal{E} \to \mathbb{R}_+$ is a \textbf{quasi-flow} if it satisfies only condition (i). A {\bf path} $P$ is a sequence of edges such that (i) the source $s$ and sink $t$ each appear in exactly one edge, and (ii) every other $n \in \mathcal{N}$ that appears in any edge appears exactly twice, once as a head and once as a tail. A path $P$ is said to pass through a node $n$ if $n$ appears in some edge in $P$, and for any node $n \in \mathcal{N}$ appearing in $P$, we let $P_n^\downarrow$ (resp.\ $P^n_\downarrow$) denote the unique edge in $P$ whose head (resp.\ tail) is $n$.\footnote{Further, we use $s_k^\uparrow(P)$ (resp. $s^k_\downarrow(P)$) to denote the first $k$ (resp. last $|P|-k$) edges in a path $P$. We will also sometimes use $s^\uparrow_n$ and $s_\downarrow^n$ to denote the initial/terminal segments of a path up to/starting from some $n \in \mathcal{N}$.} Let $\mathcal{P}$ denote the collection of all paths.\footnote{In particular, throughout the appendix we will use $\mathcal{P}$ to denote the path set of a graph, rather than the set of random choice rules as used in the body.} We say a non-negative measure $\pi \in \mathbb{R}^\mathcal{P}_+$ defines a {\bf path decomposition} of a quasi-flow $f$ if, for all $e \in \mathcal{E}$:
\[
    f(e) = \sum_{\{P \in \mathcal{P} \, :\, e \in P\}} \pi(P).
\]
A path decomposition belongs to $\Delta(\mathcal{P})$ if and only if it decomposes a flow. A map $f: \mathcal{E} \to \mathbb{R}_+$ admits at least one path decomposition if and only if it is a quasi-flow (e.g.\ \citealt{ahuja1993network}).\vspace{-2mm}

\subsubsection*{Compatibility \& Ryser Swaps}

Two paths $P_1$ and $P_2$ are $n$-{\bf compatible} if they pass through some common node $n \in \mathcal{N}$. A pair of paths $P_1', P_2'$ are $n$-{\bf conjugate} to an $n$-compatible pair $P_1,P_2$ if they are the paths obtained by following $P_1$ until $n$, then $P_2$ thereafter, and vice-versa.\footnote{That is, $P_1'$ agrees with $P_1$ up to $n$, then agrees with $P_2$ thereafter, and vice-versa.} Two path decompositions $\pi,\pi' \in \Delta(\mathcal{P})$ differ by a {\bf weighted Ryser swap} if:
\begin{equation}\label{wrs}
    \pi - \pi' = \varepsilon \,\big( \mathbbm{1}_{P_1} + \mathbbm{1}_{P_2} - \mathbbm{1}_{P_1'} - \mathbbm{1}_{P_2'} \big),
\end{equation}
where $P_1'$ and $P_2'$ are conjugate to some compatible pair $P_1,P_2$.\footnote{Formally, we refer to any vector in $\mathbb{R}^\mathcal{P}$ of the form taken by the right-hand side of \eqref{wrs} to be a weighted Ryser swap.} Note that if $\pi'$ and $\pi$ differ by a weighted Ryser swap, they are path decompositions of the same quasi-flow. We define the {\bf Ryser subspace} $\mathcal{R} \subseteq \mathbb{R}^\mathcal{P}$ as the span of all weighted Ryser swaps with $R$ denoting a typical Ryser swap.\vspace{-2mm}

\subsubsection*{Rearrangements}

Let $P_1,\ldots, P_K$ be a sequence of paths, explicitly allowing repetition. For any $n \in \mathcal{N}$, let $\Pi_n \subseteq \{1,\ldots, N\}$ denote the subset of indices corresponding to paths $P_i$ passing through $n$. A {\bf rearrangement} of $P_1,\ldots, P_K$ is a choice of permutation $\sigma_n : \Pi_n \to \Pi_n$ for each $n \in \mathcal{N}$. Any rearrangement of $P_1,\ldots, P_K$ recursively defines a new sequence of paths $P'_1,\ldots, P'_K$ via:
\[
P'_k : \quad n_0 \to \cdots \to n_M
\]
where $n_0 \equiv s$, $n_M \equiv t$, and:
\[
    n_{i+1} = \textrm{head } [P_{(\sigma_{n_i} \, \circ\, \cdots \, \circ \, \sigma_{n_0})^{-1}(k)}]^{n_i}_\downarrow,
\]
for all $0\le i \le M-1$. In other words, $P_k'$ is formed by initially following $P_{\sigma_{n_0}^{-1}(k)}$ to its first non-source node, $n_1$, then following the outgoing edge from $n_1$ belonging to $P_{(\sigma_{n_1}\, \circ \, \sigma_{n_0})^{-1}(k)}$ to $n_2$, and so forth. Two sequences are said to be {\bf rearrangement equivalent} if they can be obtained from one another by a rearrangement.\vspace{-2mm}

\subsubsection*{The Random Utility Graph}

Our main technical contribution is to characterize the sets of path decompositions of arbitrary quasi-flows, for general directed acyclic graphs.  Each of our economic results then obtains simply by applying these abstract theorems to the particular family of graphs encoding a specific stochastic choice theory. As our exposition has focused on the random utility model, we introduce its graphical representation, due to \cite{fiorini2004short}, here.  In \hyperlink{bru}{Online Appendix B}, we collect similar constructions for the various other stochastic choice theories mentioned in the text.\medskip

The random utility graph is defined by $\mathcal{N} =  2^{X}$ and $\mathcal{E}=\big \{(A,B)\ \vert \; B=A\setminus \{a\}, a \in A\big \}$, where $s = X$ and $t = \varnothing$. Any rationalizable (i.e.\ consistent with the random utility model) random choice rule $\rho$ induces a unique flow on this graph as follows:
\begin{equation}\label{bm}
    f\big(A \to A \setminus \{a\}\big) = \sum_{A \subseteq B}(-1)^{|B\setminus A|}\rho(a,B).
\end{equation}
The quantities $f\big(A \to A \setminus \{a\}\big)$ are often referred to as \emph{Block-Marschak polynomials}.\footnote{For the origin of \eqref{bm}, see \cite{block1959random}.}\medskip

\citet{fiorini2004short} shows this $f$ indeed defines a flow, i.e.\ satisfies both conditions of \eqref{flowdef} and, conversely, that every flow on this graph may be regarded as arising via \eqref{bm} from some unique, rationalizable random choice rule. Moreover, there is a bijective correspondence between the path set $\mathcal{P}$ and strict preferences on $X$, by regarding a path:
\[
    X \to X \setminus \{x_1\} \to \cdots \to X \setminus \{x_1, \ldots, x_{N-1}\} \to \varnothing
\]
as a sequence of nested lower contour sets, which uniquely defines the preference $x_1 \succ \cdots \succ x_N$. As a consequence, we may identify $\Delta(\mathcal{P})$ with $\Delta(\mathcal{L})$, and, under this identification, \cite{fiorini2004short} shows a distribution in $\Delta(\mathcal{P})$ rationalizes the data $\rho$ if and only if it is a path decomposition of the flow \eqref{bm}. Thus the partial identification properties of the random utility model are precisely equivalent to the path decompositions of flows on $(\mathcal{N}, \mathcal{E}, s,t)$.

\subsection{Graphical Results and Proofs}
\begin{lemma}[Zipper Lemma 1]\label{localswapflow}
    Consider two path decompositions $\pi,\pi' \in \Delta(\mathcal{P})$ of a quasi-flow $f$. Let $P \in \textrm{supp}(\pi)$, $\{P_u\}_{u=1}^U \subseteq \textrm{supp}(\pi')$ and, for some fixed $1 \leq \bar{k} < |P|$, suppose $s_{\bar{k}}^\uparrow(P)=s_{\bar{k}}^\uparrow(P_u)$ for each $u$.\footnote{We use $\vert P \vert$ to denote the number of edges in $P$.} Then, if $\pi(P) \leq \sum_{u=1}^U \pi'(P_u)$, there exists a finite sequence of weighted Ryser swaps $(c_v R_v)_{v=1}^V$ and a finite set of paths $\{P_w\}_{w=1}^W$ satisfying $s_{\bar{k}+1}^\uparrow(P)=s_{\bar{k}+1}^\uparrow(P_w)$ for all $w$, such that:
    \[
        \pi(P) \leq \sum_{w=1}^W \left(\pi'(P_w) + \sum_{v=1}^V c_v \, \langle R_v, e_{P_w}\rangle\right),
     \]
    where $e_{P_w}$ denotes the $P_w$-th standard Euclidean basis vector of $\mathbb{R}^\mathcal{P}$.
\end{lemma}

\begin{proof}
    Let $(t_{\bar{k}+1},h_{\bar{k}+1})$ denote the $(\bar{k}+1)$-st edge of the path $P$. First observe that each path in $\{P_u\}_{u=1}^U \subseteq \textrm{supp}(\pi')$ passes through node $t_{\bar{k}+1}$ since $t_{\bar{k}+1}$ is the head of the $\bar{k}$-th edge in $P$ and $s_{\bar{k}}^\uparrow(P)=s_{\bar{k}}^\uparrow(P_u)$ for each $u$. Now, since $\pi$ and $\pi'$ are path decompositions of the same quasi-flow, we have by hypothesis:
    \[
    \sum_{P'\in \mathcal{P}}\pi'(P')\mathbbm{1}_{\{(t_{\bar{k}+1},h_{\bar{k}+1}) \in P'\}}=\sum_{P'\in \mathcal{P}}\pi(P')\mathbbm{1}_{\{(t_{\bar{k}+1},h_{\bar{k}+1}) \in P'\}} \geq \pi(P),
    \]
    i.e.\ both $\pi$ and $\pi'$ put at least $\pi(P)$ mass on paths which pass through $(t_{\bar{k}+1},h_{\bar{k}+1}) \in P$.

    Let $\{P_v\}_{v=1}^V$ denote those paths in $\textrm{supp}(\pi')$ that contain $(t_{\bar{k}+1},h_{\bar{k}+1})$ as an edge. Since each $P_v$ and each $P_u$ passes through node $t_{\bar{k}+1}$, each pair $(P_u,P_v)$ is compatible at $t_{\bar{k}+1}$. Further, as $s_{\bar{k}}^\uparrow(P)=s_{\bar{k}}^\uparrow(P_u)$ for each $u$ and $(t_{\bar{k}+1},h_{\bar{k} +1})\in P_v\cap P$ for each $v$, at least one path in the $t_{\bar{k}+1}$-conjugate of each $(P_u,P_v)$ agrees with $P$ on $s_{\bar{k}+1}^\uparrow$. Let $\theta$ be the minimum between (i) the total amount of mass $\pi'$ puts on paths in $\{P_u\}_{u=1}^U$ and (ii) the total mass $\pi'$ puts on paths in $\{P_v\}_{v=1}^V$. By the preceding discussion, $\theta \geq \pi(P)$. For each pair of paths $(P_u,P_v)$, let
    \[
    c_{uv}=\left( \frac{\pi'(P_u)}{\sum_{u'=1}^U \pi'(P_{u'})}\right)\left( \frac{\pi'(P_v)}{\sum_{v'=1}^V \pi'(P_{v'})}\right)\theta.
    \]
    By construction, we get that $\sum_{u=1}^U \sum_{v=1}^V c_{uv}=\theta$. Further, $\sum_{u=1}^Uc_{uv}=\pi'(P_v)\frac{\theta}{\sum_{v'=1}^V \pi'(P_{v'})}\leq \pi'(P_v)$, and analogously for $P_u$.  Now let $R_{uv}$ be the Ryser swap which puts $-1$ mass on paths $P_u$ and $P_v$ and $+1$ mass on their $t_{\bar{k}+1}$-conjugate. Further, by previous arguments, the signed measure resulting from $\sum_{u=1}^U \sum_{v=1}^V c_{uv}R_{uv}$ puts $\theta$ mass on paths which agree with $P$ on $S_{\bar{k}+1}^\uparrow(P)$. Hence we obtain a finite sequence of weighted Ryser swaps which shifts $\theta$ mass onto paths $\{P_w\}_{w=1}^W$ satisfying $s_{\bar{k}+1}^\uparrow(P)=s_{\bar{k}+1}^\uparrow(P_w)$, which gives:
    $$\sum_{w=1}^W \left(\pi'(P_w) + \sum_{v=1}^V c_v \, \langle R_v, e_{P_w}\rangle \right) \geq \theta \geq \pi(P)$$
    as desired while maintaining non-negativity of the resulting measure.
\end{proof}

\begin{lemma}[Zipper Lemma 2]\label{swapsupportflow}
        Suppose $\pi$ and $\pi'$ are path decompositions of some quasi-flow $f$. Fix a path $P$ in the support of $\pi$. Then there exists a finite sequence of weighted Ryser swaps with positive coefficients $\{c_iR_i\}$, $c_i > 0$, such that $\pi' + \sum_i c_i R_i = \pi''$, $\pi''(P)\geq \pi(P)$, and $\pi''$ is a path-decomposition.
\end{lemma}

\begin{proof}
    Fix $P$ and denote its first edge $(s,n_1)$. Since $\pi$ and $\pi'$ are path decompositions of $f$, 
    \[
    \sum_{P' \in \mathcal{P}}\pi(P) \mathbbm{1}_{\{(s,n_1) \in P'\}}=\sum_{P' \in \mathcal{P}}\pi'(P') \mathbbm{1}_{\{(s,n_1) \in P'\}}\geq \pi(P) > 0.
    \]
    As a consequence, the set of paths in $\textrm{supp}(\pi')$ which include $(s,n_1)$ is non-empty; enumerate it $\{P_u\}_{u=1}^U$. Then by the above we have $\pi(P) \leq \sum_{u=1}^U \pi'(P_u)$, and hence the requirements of \autoref{localswapflow} are satisfied. However, note that the enumerated set of paths $\{P_w\}_{w=1}^W$ whose existence \autoref{localswapflow} asserts \emph{also} satisfy the requirements of \autoref{localswapflow} (for the same fixed choice of $P$), but this time for $s_2^\uparrow(P)$. Repeating this process $\vert P \vert$ times yields a path decomposition $\pi''$ differing from $\pi$ by a sum of weighted Ryser swaps, and an enumerated set of paths $\{P_l\}_{l=1}^L$ belonging to the support of $\pi''$ such that (i) $s_{\vert P \vert}^\uparrow(P) = s_{\vert P \vert}^\uparrow (P_l)$ and (ii) $\pi(P) \le \sum_{l=1}^L \pi''(P_l)$, but (i) implies $L = 1$ and $P_l = P$, hence (ii) reduces to $\pi''(P) \ge \pi(P)$ as desired.
\end{proof}

\begin{lemma}\label{MainLemmaflow}
    Both $\pi$ and $\pi'$ are path decompositions of the same quasi-flow $f$ if and only if there exists a finite sequence of Ryser swaps with positive coefficients $\{c_iR_i\}$, $c_i > 0$, such that $\pi + \sum_i c_i R_i = \pi'$.
\end{lemma}

\begin{proof}
    $(\Longleftarrow)$: Suppose $\pi + \sum_i c_i R_i = \pi'$, where $\{c_i R_i\}$ is a finite sequence of weighted Ryser swaps. If two path decompositions differ by a single weighted Ryser swap, they are necessarily path decompositions of the same quasi-flow. By induction, we obtain $\pi'$ and $\pi$ are path decompositions of the same quasi-flow $f$. 
    
    $(\Longrightarrow)$: Toward the converse, we provide an explicit algorithm to construct $\pi$ from $\pi'$ via a sequence of weighted Ryser swaps.
    \begin{enumerate}
        \item Initialize by enumerating the set of paths in the support of $\pi'$ via $i \in \{1, \dots, I\}$ and set $i=1$. Set $\pi_1=\pi$ and set $\pi'_1=\pi'$. 
        \item $\pi_i$ and $\pi'_i$ are path decompositions of the quasi-flow $f_i$.\footnote{Here, $f_1=f$ and for all $i>1$, $f_i(e)=\sum_{\{P \in \mathcal{P} \, :\, e \in P\}} \pi_i(P)$.} As such, by \autoref{swapsupportflow}, we obtain a path decomposition  $\pi''_i$ of the same quasi-flow as $\pi_i, \pi'_i$, but satisfying $\pi_i''(P_i) \ge \pi_i(P_i)$.
        \item Set $i=i+1$. Set $\pi'_{i+1}(P)=\pi''_{i}(P)-\pi(P)\mathbbm{1}_{\{P=P_{i}\}}$, and $\pi_{i+1}(P)=\pi_{i}(P)-\pi(P)\mathbbm{1}_{\{P=P_{i}\}}$.
        \item If $\pi_i=\mathbf{0}$, terminate the algorithm. If not, return to step 2.
    \end{enumerate}
    At each step, the algorithm applies only finitely many weighted Ryser swaps. As it terminates with $\pi_i=\mathbf{0}$, it entails a finite sequence of weighted Ryser swaps which transforms $\pi'$ into $\pi$, as desired.
\end{proof}

\subsubsection*{Proof of \autoref{ryserspacetheorem} and \autoref{conjugatelemma}}

\autoref{ryserspacetheorem} follows from applying the following abstract result to the random utility graph defined at the start of this appendix. \autoref{conjugatelemma} then  follows as an immediate corollary.

\begin{theorem}\label{thm:flowryserspacetheorem}
    A pair $\pi, \pi' \in \Delta(\mathcal{P})$ are path decompositions of the same quasi-flow if and only if $\pi - \pi' \in \mathcal{R}$. In particular, for any $\pi \in \mathcal{M}\subseteq \Delta(\mathcal{P})$, the set of decompositions of the same flow as $\pi$ corresponds to $\big(\pi + \mathcal{R}\big) \cap \mathcal{M}$.
\end{theorem}

\begin{proof}
    The first claim follows directly from \autoref{MainLemmaflow}. For the latter, observe that $\{\sum_i c_i R_i| c_i \in \mathbb{R}\}$ defines $\mathcal{R}$. So $\pi + \mathcal{R}$ is the set $\{\pi+\sum_i c_i R_i| c_i \in \mathbb{R}_{++}\}$, as the set of (unweighted) Ryser swaps is closed under multiplication by $-1$. Thus again by \autoref{MainLemmaflow} the claim follows.
\end{proof}

\subsubsection{Graphical Support Restrictions}

\begin{theorem}\label{thm:GraphSupp}
    Let $P_1,\dots,P_K$ and $P_1',\dots,P_K'$ be finite sequences of paths, allowing for repetition. $P_1,\dots,P_K$ and $P_1',\dots,P_K'$ are rearrangement equivalent if and only if the sums:
    \[
        \sum_{k=1}^K \mathbbm{1}_{P_k} \quad \textrm{ and } \quad \sum_{k=1}^K \mathbbm{1}_{P_k'}
    \]
    differ by a finite sum of unweighted Ryser swaps.
\end{theorem}

Before proving \autoref{thm:GraphSupp}, we require a preliminary lemma.

\begin{lemma}\label{lem:orderedNodes}
    For any directed acyclic graph $(\mathcal{N},\mathcal{E},s,t)$, there is an enumeration of the set of nodes from $1$ to $|\mathcal{N}|$ such that, for each $i$, there is no $j \geq i$ such that $(n_j,n_i)\in \mathcal{E}$.
\end{lemma}

\begin{proof}
    Let $d: \mathcal{N} \to \mathbb{Z}_+$ denote the function which assigns to each node $n \in \mathcal{N}$ the length (i.e.\ number of edges) of the longest sequence of edges in $\mathcal{E}$ of the form: $s \to \cdots \to n$ connecting $s$ to $n$. As $(\mathcal{N},\mathcal{E},s,t)$ is finite and acyclic, $d$ is well-defined.\medskip

    Let $\succeq_d$ denote the weak order on $\mathcal{N}$ represented by $d$, and $\succeq^*$ any linear order extension of its asymmetric component. Let $n_i$ denote the unique node such that $\big\vert \{n \in \mathcal{N}: n_i \succeq^* n\}\big\vert = i$. It is straightforward to verify this gives an ordering $n_1,\ldots, n_{\vert \mathcal{N}\vert}$.  We claim this ordering has the desired monotonicity properties.\medskip

    To see this, suppose for sake of contradiction that $(n_j, n_i) \in \mathcal{E}$, where $j \ge i$. Since $(\mathcal{N},\mathcal{E},s,t)$ is acyclic, $j > i$. For any sequence of edges $s \to \cdots \to n_j$, by appending the edge $ n_j \to n_i$ we obtain a sequence terminating at $n_i$, hence $d(n_i) > d(n_j)$ and $n_i \succ_d n_j$, and therefore $n_i \succ^* n_j$, which implies $i > j$, a contradiction.
\end{proof}   

We now move to our proof of \autoref{thm:GraphSupp}.

\begin{proof}
    To begin, suppose that $P_1,\dots,P_K$ and $P_1',\dots,P_K'$ are rearrangement equivalent. Choose some enumeration of the nodes of our graph satisfying the conditions of \autoref{lem:orderedNodes}. For a path $P$ that passes through node $n$, let $s_n^\uparrow(P)$ denote the subset of edges in $P$ that precede node $n$ (and thus ends with an edge of form $(m,n)$). Similarly, for a path $P$ which passes through node $n$, let $s^n_\downarrow(P)$ denote the subset of edges in $P$ that come after node $n$ (and thus begins with and edge of form $(n,m)$). We now propose an algorithm that takes a rearrangement and constructs Ryser swaps that induce the rearrangement.
    \begin{enumerate}
        \item Initialize at $j=1$, and define the signed measure $\mu_0 = \sum_{k=1}^K \mathbbm{1}_{P_k}$, and sequence of paths $P^0_k \equiv P_k$, for all $ k = 1,\ldots, K$
        \item For node $n_j$, let $\Pi_{j} \subseteq \{1,\ldots, K\}$ consist of the indices of the paths in $P^{j-1}_1,\dots,P^{j-1}_k$ which pass through $n_j$. For any $i \in \Pi_j$, let $R^j_i = \mathbbm{1}_{\{\tilde{P}_i, \tilde{P}_{\sigma_{n_j}^{-1}(i)}\}}- \mathbbm{1}_{\{P_i, P_{\sigma_{n_j}^{-1}(i)}\}}$, where $\tilde{P}_i=s_{n_j}^\uparrow(P_i)s_\downarrow^{n_j}(P_{\sigma_{n_j}^{-1}(i)})$ and $\tilde{P}_{\sigma_{n_j}^{-1}(i)}=s_{n_j}^\uparrow(P_{\sigma_{n_j}^{-1}(i)})s_\downarrow^{n_j}(P_i)$. Note every $R_i^j$, by definition of a rearrangement, is a valid Ryser swap. Define $\mu_j = \mu_{j-1} + \sum_{i \in \Pi_j} R^j_i$.This creates a new, $\mathbb{Z}_+$-valued signed measure in $\mathbb{R}^\mathcal{P}$ which we may identify with a sequence of paths $P^j_1, \ldots, P^j_K$, where:
        \[
            P^j_k = \begin{cases}
                P^{j-1}_k & \textrm{ if } k \not \in \Pi_j\\
                s_{n_j}^\uparrow(P^j_k)s_\downarrow^{n_j}\big(P^j_{\sigma_{n_j}^{-1}(k)}\big) & \textrm { else.}
            \end{cases}
        \]
        \item Set $j=j+1$. If $n_j=t$, terminate the algorithm. If $n_j \neq t$, return to step 2.
    \end{enumerate}

By construction, the algorithm terminates with a set of paths $P_1^{\vert \mathcal{N}\vert},\dots,P_k^{\vert \mathcal{N}\vert}$ such that, if $P^{\vert \mathcal{N}\vert}_k$ is written as $s \to n_1 \to \cdots \to t$, we have $ n_{i+1} = \textrm{head}\big[P_{(\sigma_{n_i}\, \circ \, \cdots \, \circ \, \sigma_{s})^{-1}(k)} \big]^{n_i}_\downarrow$ and hence $P_1^{\vert \mathcal{N}\vert},\dots,P_k^{\vert \mathcal{N}\vert}$ is precisely the sequence $P_1', \ldots P'_K$. Therefore, $\sum_{k=1}^K \mathbbm{1}_{P_k'} = \sum_{k=1}^K \mathbbm{1}_{P_k} + \sum_{j=1}^{\vert \mathcal{N} \vert}\sum_{i \in N_j} R_i^j$ as desired.\medskip

Conversely, suppose $\sum_{k=1}^K \mathbbm{1}_{P_k'} - \sum_{k=1}^K \mathbbm{1}_{P_k} = \sum_{l=1}^L R_l$. We now show their underlying sequences $P_1,\dots,P_K$ and $P_1',\dots,P_K'$ are rearrangement equivalent by induction on $L$. Suppose first that $L=1$. By relabeling, without loss we may suppose the sequences $P_1,\ldots, P_K$ to $P_1', \ldots, P_K'$ are identical but for two paths in each, $P_i, P_j$ and $P'_{i'}, P'_{j'}$, and $R_1$ must correspond to swapping the outgoing terminal segments $s^{n^*}_\downarrow(P_i)$ and $s^{n^*}_\downarrow(P_j)$ to yield $P'_{i'} = s_{n^*}^\uparrow(P_i)s^{n^*}_\downarrow(P_j)$ and $P'_{j'} = s_{n^*}^\uparrow(P_j)s^{n^*}_\downarrow(P_i)$. In particular, $s_{n^*}^\uparrow(P'_{i'}) = s_{n^*}^\uparrow(P_i)$ and $s_{n^*}^\uparrow(P'_{j'}) = s_{n^*}^\uparrow(P_j)$. Then define a rearrangement where $\sigma_s$ fixes every index except $i,j, i', j'$ and maps $i \leftrightarrow i'$ and $j \leftrightarrow j'$, $\sigma_{n^*}$ fixes every element in $\Pi_{n^*}$ except $i,j$ and maps $i \leftrightarrow j$, and $\sigma_n$ is the identity for all $n \in \mathcal{N} \setminus \{s,n^*\}$. By construction, this defines a rearrangement taking $P_1,\ldots, P_K$ to $P_1', \ldots, P_K'$.

Now, suppose that for $L \le \bar{L}$ the claim holds, and let $\sum_{k=1}^K \mathbbm{1}_{P_k'}$ and $\sum_{k=1}^K \mathbbm{1}_{P_k}$ differ by a sum of $\bar{L}+1$ weighted Ryser swaps. Consider first the case in which there exists some $R_l$, without loss $R_{\bar{L}+1}$, such that $\sum_{k=1}^K \mathbbm{1}_{P_k} + R_{\bar{L}+1} \ge 0$. Then $\sum_{k=1}^K \mathbbm{1}_{P_k} + R_{\bar{L}+1}$ differs from $\sum_{k=1}^K \mathbbm{1}_{P'_k}$ by a sum of $\bar{L}$ unweighted Ryser swaps, and hence identifying $\sum_{k=1}^K \mathbbm{1}_{P_k} + R_{\bar{L}+1}$ with a sequence $\tilde{P}_1, \ldots, \tilde{P}_K$, by hypothesis this sequence is rearrangement equivalent to $P'_1,\ldots, P'_K$. Moreover, since $1 \le \bar{L}$, we obtain that similarly $P_1,\ldots, P_K$ is rearrangement equivalent to $\tilde{P}_1, \ldots, \tilde{P}_K$. Denote these rearrangements by $\{\sigma^1_n\}_{n \in \mathcal{N}}$ and $\{\sigma^2_n\}_{n \in \mathcal{N}}$ respectively.  Then defining, for each $n \in \mathcal{N}$, $\sigma_n = \sigma^1_n \circ \sigma^2_n$ defines a rearrangement under which $P_1, ,\ldots, P_K$ and $P_1', \ldots, P_K'$ are equivalent.

On the other hand, suppose for every $R_l$, the sum $\sum_{k=1}^K \mathbbm{1}_{P_k} + R_l \not \ge \mathbf{0}$. Define $\mathcal{Q}$ as the set of paths which receive negative mass under $\sum_{k=1}^K \mathbbm{1}_{P_k} + R_{\bar{L}+1}$. Then $
\sum_{k=1}^K \mathbbm{1}_{P_k} + R_{\bar{L}+1} + \sum_{P \in \mathcal{Q}} \mathbbm{1}_P
$ and $
\sum_{k=1}^K \mathbbm{1}_{P_k} + \sum_{P \in \mathcal{Q}} \mathbbm{1}_P 
$
differ by a single Ryser swap, and the former differs from $
    \sum_{k=1}^K \mathbbm{1}_{P'_k} + \sum_{P \in \mathcal{Q}} \mathbbm{1}_P
$ by $\bar{L}$ Ryser swaps. By the preceding argument, the sequence representations of these last two expressions are rearrangement equivalent.  We may regard the path(s) in $\mathcal{Q}$ as the last path(s) in each sequence; if this rearrangement does not fix the paths corresponding to these last indices, then the path(s) in $\mathcal{Q}$ appearing in the non-prime sequence correspond to some $P'_i$ or $P'_i$ and $P_j'$.  In such a case, define a new rearrangement that is the identity at every node except $s$, and which at $s$ maps (respectively) $i \leftrightarrow K+1$ or $i \leftrightarrow K+1$ and $j \leftrightarrow K+2$, and fixes every other index. Composing this new rearrangement with the former then gives a rearrangement between the augmented sequences which fixes the terms in the sequence corresponding to $\mathcal{Q}$, and hence restricts to a rearrangement between $P_1,\ldots, P_K$ and $P_1', \ldots, P_K'$ as desired.
\end{proof}

Given a directed acyclic graph, its set of paths $\mathcal{P}$, and some subset $\mathcal{P}'$, let $\mathbb{R}^\mathcal{P}_{\mathcal{P}'}$ denote the set of vectors in $\mathbb{R}^\mathcal{P}$ which take on value zero in the dimensions indexed by paths $P \in \mathcal{P}\setminus\mathcal{P}'$. By minor abuse of notation, let $\Delta(\mathcal{P}')$ denote the set of path decompositions whose support is a subset of $\mathcal{P}'\subseteq \mathcal{P}$.

\begin{theorem}\label{thm:GraphSupportRestrictions}
    Fix a graph $(\mathcal{N},\mathcal{E},s,t)$. The following are equivalent.
    \begin{enumerate}
        \item For each $\pi\in \Delta(\mathcal{P}')$, there is no path decomposition $\pi'\neq\pi$ with $\pi'\in \Delta(\mathcal{P}')$ that induces the same flow as $\pi$.
        \item For all $\pi \in \Delta(\mathcal{P'})$,
        \[
            \big(\pi + \mathcal{R}\big) \cap \Delta(\mathcal{P}') = \{\pi\}.
        \]
        \item $\mathcal{P}'$ contains no distinct, rearrangement-equivalent sequences of paths.
        \item For every finite sequence of \emph{unweighted} Ryser swaps $\{R_i\}_{i=1}^K$ (allowing for repetition):
        \[
            \sum_{i=1}^K R_i \in \mathbb{R}^\mathcal{P}_{\mathcal{P}'}  \iff \sum_{i=1}^K R_i = \mathbf{0}.
        \]
    \end{enumerate}
\end{theorem}

\begin{proof}
    As $\Delta(\mathcal{P}')\subseteq \Delta(\mathcal{P})$, the equivalence between (1) and (2) is an immediate consequence of \autoref{thm:flowryserspacetheorem}. We now show (3) implies (1). By \autoref{rationalLem}, (1) holds if and only the restriction $\mathcal{M} = \Delta(\mathcal{P}')\cap \mathbb{Q}^\mathcal{P}$ is identifying. Suppose, for purposes of contraposition, that $\Delta(\mathcal{P}')\cap \mathbb{Q}^\mathcal{P}$ is not identifying. Then there exist distinct $\pi,\pi' \in \Delta(\mathcal{P}')\cap \mathbb{Q}^\mathcal{P}$ which decompose the same flow. By \autoref{MainLemmaflow}, $\pi$ and $\pi'$ differ by a finite sum of Ryser swaps, $\pi-\pi'=\sum_{i=1}^KR_i.$ Since both $\pi$ and $\pi'$ take only rational values, by multiplying by a sufficiently large positive integer, we obtain $\mathbb{Z}_+$-valued signed measures  $\bar{\pi}$, $\bar{\pi}'$. We may identify these signed measures with sequences of paths $P_1,\dots,P_K$ and $P_1',\dots,P_K'$, where $\bar{\pi} = \sum_{k=1}^K \mathbbm{1}_{P_k}$ and $\bar{\pi}' = \sum_{k=1}^K \mathbbm{1}_{P_k'}$. As these signed measures differ by a sum of Ryser swaps (by repeatedly summing copies of $\sum_i R_i$), by \autoref{thm:GraphSupp}, $\pi_N$ and $\pi_N'$ being rearrangement-equivalent. Thus  $\neg(1)$ implies $\neg (3)$.

    We now show that (4) implies (3). Once more, by means of contraposition, suppose $P_1,\dots,P_k$ and $P_1',\dots,P_k'$ are distinct, rearrangement-equivalent sequences of paths. Then by \autoref{thm:GraphSupp}, $\sum_{i=1}^N R_i = \sum_{k=1}^K \mathbbm{1}_{P_k} - \mathbbm{1}_{P_k'} \neq \mathbf{0}$ is a finite sum of unweighted Ryser swaps. Furthermore, since $P_1,\dots,P_k$ and $P_1',\dots,P_k'$ contain only paths in $\mathcal{P}'$, the sum $\sum_{i=1}^K R_i$ is also supported on $\mathcal{P}'$. Thus $\sum_{i=1}^k \mathbbm{1}_{P_i} - \sum_{i=1}^k \mathbbm{1}_{P_i'}=\sum_{i=1}^KR_i\neq \mathbf{0}$ and yet $\sum_{i=1}^k \mathbbm{1}_{P_i} - \sum_{i=1}^k \mathbbm{1}_{P_i'}=\sum_{i=1}^KR_i\in \mathbb{R}^\mathcal{P}_{\mathcal{P}'}$. Thus $\neg (3)$ implies $\neg(4)$.

    Finally, we show (1) implies (4). To this end, suppose (4) is false. Thus there exists some finite sequence of unweighted Ryser swaps such that $\sum_{i=1}^KR_i\neq \mathbf{0}$ with $\sum_{i=1}^KR_i\in \mathbb{R}^\mathcal{P}_{\mathcal{P}'}$. Since each Ryser swap has two entries equal to one, two entries equal to negative one, and all other entries equal to zero, $\sum_{i=1}^KR_i$ puts equal mass on its negative entries and its positive entries. We can then renormalize and treat the mass on the negative entries as a path decomposition and the mass on the positive entries as a path decomposition. By construction, these differ by a finite sum of (potentially weighted) Ryser swaps. By \autoref{MainLemmaflow}, they induce the same flow and, by construction, they have distinct supports and are thus distinct. It immediately follows that condition (1) of \autoref{thm:GraphSupportRestrictions} fails. Thus together we have shown that $\neg (1) \implies \neg (3) \implies \neg (4) \implies \neg (1)$.
\end{proof}

\begin{corollary}\label{cor:LinInd}
    For any $\mathcal{P}' \subseteq \mathcal{P}$, the flows $\{\mathbbm{1}_{\{e\in P\}}\}_{P\in \mathcal{P}'}$ are linearly independent if and only if $\mathcal{P}'$ contains no distinct, rearrangement-equivalent sequences of paths.
\end{corollary}

\begin{proof}
    $(\Longleftarrow)$: Suppose, by way of contraposition, that these flows are linearly dependent. Thus there exist scalars $\{w_P\}_{P\in\mathcal{P}'}$, not all zero, such that $\sum_{P\in \mathcal{P}'}w_P\mathbbm{1}_{\{e\in P\}}=0$. Define mass functions $\pi_+(P)=\max\{0,w_P\}$ and let $\pi_-(P)=\max \{0,-w_P\}$ in $\mathbb{R}^\mathcal{P}$. As every path leaves the source from a unique edge, and $\pi_+$ and $\pi_-$ put the same total mass on each edge, $\pi_+$ and $\pi_-$ have equal total mass. Further, as $\sum_{P\in \mathcal{P}}[\pi(P)_+\mathbbm{1}_{\{e\in P\}}-\pi_-(P)\mathbbm{1}_{\{e\in P\}}]=0$, i.e.\ the zero quasi-flow, we obtain that  $\pi_+$ and $\pi_-$ are path decompositions of the same quasi-flow, which is supported only on the edge set of paths in $\mathcal{P}'$. By \autoref{thm:GraphSupportRestrictions}, $\mathcal{P}'$ then contains a pair of distinct, rearrangement-equivalent sequences, as desired.

    $(\Longrightarrow):$ Conversely, suppose $\{\mathbbm{1}_{\{e\in P\}}\}_{P\in \mathcal{P}'}$ are linearly independent.  Thus these vectors span a simplex, and any path decomposition is necessarily unique. Then, by equivalence of (1) and (3) in \autoref{thm:GraphSupportRestrictions}, $\mathcal{P}'$ contains no pair of distinct, rearrangement-equivalent sequences.
\end{proof}

\begin{theorem}\label{thm:Graphextremepoints}
    Fix a graph $(\mathcal{N},\mathcal{E},s,t)$. As a slight abuse of notation, let $\Delta(\mathcal{P}')$ denote the set of path decompositions whose support is a subset of $\mathcal{P}'\subseteq \mathcal{P}$. The following are equivalent for $\pi \in \Delta(\mathcal{P}')$.
    \begin{enumerate}
        \item $\pi$ is an extreme point of $\Delta(\mathcal{P}') \cap \big(\pi + \mathcal{R}\big)$.
        \item The set $\{\mathbbm{1}_{\{e\in P\}}:  \pi(P) > 0\}_{P \in \mathcal{P}'}$ is linearly independent.
        \item $\textrm{supp}(\pi)$ contains no distinct, rearrangement-equivalent sequences of paths.
        \item For every finite sequence of \emph{unweighted} Ryser swaps $\{R_i\}_{i=1}^K$ (allowing for repetition):
        \[
            \sum_{i=1}^K R_i \in \mathbb{R}^\mathcal{P}_{\textrm{supp}(\pi)}  \iff \sum_{i=1}^K R_i = \mathbf{0}.
        \]
    \end{enumerate}
\end{theorem}

\begin{proof}
    The set $\Delta(\mathcal{P}') \cap \big(\pi + \mathcal{R}\big)$ defines a moment set in the sense of \citet{winkler1988extreme}, i.e.\ it is a set of probability measures which satisfy some auxiliary, finite system of (linear) equality constraints. By Proposition 2.1.a of \citet{winkler1988extreme}, $\pi$ is an extreme point of $\Delta(\mathcal{P}') \cap \big(\pi + \mathcal{R}\big)$ if and only if  the set $\{\mathbbm{1}_{\{e\in P\}}:  \pi(P) > 0\}_{P \in \mathcal{P}'}$ is linearly independent. This shows the equivalence of (1) and (2). The equivalence between (2), (3), and (4) follows directly from \autoref{cor:LinInd} and \autoref{thm:GraphSupportRestrictions}.
\end{proof}

\subsubsection*{Proof of \autoref{thm:extremepoints} and \autoref{thm:braids}}

Our \autoref{thm:extremepoints} corresponds to the equivalence between ($1$) and ($3$) in \autoref{thm:Graphextremepoints} applied to the random utility graph. Similarly, \autoref{thm:braids} follows from the equivalence between conditions ($1$) and ($3$) in \autoref{thm:GraphSupportRestrictions} applied to the random utility graph.

\subsubsection{Graphical Ordered Results}

Fix a graph and let $\pi$ be a path decomposition of quasi-flow $f$. Suppose that $\trianglerighteq$ is a linear order over the edges of the fixed graph. We say that $\pi$ is \textbf{swap-progressive} (with respect to $\trianglerighteq$) if $\textrm{supp}(\pi)$ can be ordered $P_1,\dots,P_N$ such that, whenever $P_i$ and $P_j$ are compatible at some node $n$, then $i > j$ implies that $(n,m_i) \trianglerighteq (n,m_j)$ where $(n,m_i)$ and $(n,m_j)$ are the edges of $P_i$ and $P_j$, respectively, leaving node $n$.

\begin{theorem}\label{thm:OrderedFlows}
    Fix a graph, a quasi-flow function $f$, and a linear order $\trianglerighteq$ over the edges of our graph. There is some swap-progressive path decomposition of $f$ and it is unique.
\end{theorem}

\begin{proof}
    We first show uniqueness. Suppose, toward a contradiction, that the quasi-flow $f$ admits two distinct swap-progressive path decompositions,  $\pi_1$ and $\pi_2$. Consider the zero-sum signed measure $\pi_1 - \pi_2$. Since $\pi_1$ and $\pi_2$ decompose the same quasi-flow, $\pi_1-\pi_2 = \sum_{i=1}^N c_i R_i$ where $c_i R_i \in \mathcal{R}$ for all $i = 1,\ldots, N$. Define the signed measures $\pi_+, \pi_- \in \mathbb{R}^\mathcal{P}$ via  $\pi_+(P)=\max\{0,\pi_1(P)-\pi_2(P)\}$ and $\pi_-(P)=\max\{0,\pi_2(P)-\pi_1(P)\}$. As $\pi_1,\pi_2 \ge \mathbf{0}$, $\textrm{supp}(\pi_+) \subseteq \textrm{supp}(\pi_1)$ and similarly  $\textrm{supp}(\pi_-) \subseteq \textrm{supp}(\pi_2)$.  By restriction of the ordering of the supports of $\pi_1$ and $\pi_2$, the supports of $\pi_+$ and $\pi_-$ still satisfy the definition of swap-progressivity and moreover:
    \[
        \sum_{P \in \textrm{supp}(\pi_+)} \pi_+(P) \mathbbm{1}_{\{e \in P\}} = \sum_{P' \in \textrm{supp}(\pi_-)} \pi_-(P') \mathbbm{1}_{\{e \in P'\}}
    \]
    for every $e \in \mathcal{E}$ and thus are path decompositions of the same quasi-flow.

    Let $P_+$ be the highest ordered path in the support of $\pi_+$. As $\pi_-$ equals $\pi_+$ modulo a finite sequence of weighted Ryser swaps, there is some path $P_-$ in the support of $\pi_-$ and some number $k>0$ such that $s_k^\uparrow(P_-)= s_k^\uparrow(P_+)$. Let $P_-$ be the highest ordered path in the support of $\pi_-$ such that $s_k^\uparrow(P_-)= s_k^\uparrow(P_+)$ for some $k>0$. Further, $s^k_\downarrow(P_-)\neq s^k_\downarrow(P_+)$ for the corresponding $k$ by disjointedness of the supports of $\pi_+$ and $\pi_-$. Let $e_+$ and $e_-$ denote the first edges at which $P_+$ and $P_-$ disagree. Since $\pi_-$ can be reached from $\pi_+$ via a sequence of weighted Ryser swaps, we get that there is some path, $P_+'$, in the support of $\pi_+$ such that $e_- \in P_+'$. By $P_+$ being the highest ordered path in the support of $\pi_+$, we get that $e_+ \trianglerighteq e_-$ in our exogenous ordering. Once again, by $\pi_+$ being $\pi_-$ plus a sequence of Ryser swaps, we get that there is some path in the support of $\pi_-$, $P_-'$, which has $e_+ \in P_-'$. There are two cases.
    \begin{enumerate}
        \item $s_{k'}^\uparrow(P'_-)= s_{k'}^\uparrow(P_+)$ for some $k'>0$. \\
        In this case, we have $P_- > P_-'$ by $P_-$ being the highest ordered path for which $s_k^\uparrow(P_-)= s_k^\uparrow(P_+)$ for some $k>0$. However, we have that $e_+ \trianglerighteq e_-$, which gives us $P_-' > P_-$ as both $e_+$ and $e_-$ originate from the same node. This is a contradiction to swap-progressiveness.
        \item $s_{k}^\uparrow(P'_-)\neq s_{k}^\uparrow(P_+)$ for all $k>0$. \\
        Let $e_+'$ and $e_-'$ denote the edges where $P_-'$ and $P_+$ first disagree. We now have $e_+' \in P_+ \cap P_-$ (by $s_{k}^\uparrow(P_-)= s_{k}^\uparrow(P_+)$ for some $k>0$ and $e_+'$ being the first edge of $P_+$) and $e_-' \in P_-'$. Once again, applying the Ryser swap theorem, we get that there is some path in the support of $\pi_+$, $P_+'$, such that $e_-' \in P_+'$. This, plus $P_+$ being the highest ordered path, gives us that $e_+' \trianglerighteq e_-'$. However, $e_+ \trianglerighteq e_-$ means that $P_-' > P_-$ and $e_+' \trianglerighteq e_-'$ means that $P_- > P_-'$. This is a contradiction to swap-progressiveness.
    \end{enumerate}
    By contradiction in our two cases, we get that there cannot be two different swap-progressive path decompositions for the same flow function, and so we are done.

    We now show that every quasi-flow function can be induced by a swap-progressive path decomposition.  We present an algorithm which takes as input a quasi-flow and returns the unique swap-progressive path decomposition.
    \begin{enumerate}
        \item Initialize at $k=0$. Set $f_0(e)=f(e)$ and $\pi_0(P)=0$ for all $P \in \mathcal{P}$.
    \item Now, let $k = k+1$.
        \begin{enumerate}
             \item Set $l=1$, $n_1=s$, and $\tilde{P}_k=\emptyset$. Let $e$ be the highest ranked edge according to $\trianglerighteq$ at node $n_l$ such that $f_k(e)>0$. Set $\tilde{P}_k=\tilde{P}_k\cup\{e\}$ and $m=n_{l+1}$ where $m$ is the head of $e$.
        \item Set $l=l+1$. By property (i) of \eqref{flowdef}, at every interior node, we have that there is some edge $e$ whose tail is $n_l$ such that $f_k(e)>0$.\footnote{We note that the algorithm starts with property (i) of \eqref{flowdef} holding. At every iteration of the algorithm, we subtract out a vector proportional to a flow along a single path, thus preserving property (i).} Let $e$ be the highest ranked edge according to $\trianglerighteq$ whose tail is $n_l$ satisfying $f_k(e)>0$. Set $\tilde{P}_k=\tilde{P}_k\cup\{e\}$ and $m=n_{l+1}$ where $m$ is the head of $e$.
        \item If $n_{l+1} = t$, proceed to the Step 3. Otherwise, return to 2.(b).     
        \end{enumerate}
       
        \item We have now defined a path $P_k = \tilde{P}_k$. Set $\pi_k(P_k)=\min_{e \in P_k} f_k(e)$.
        \item For $e \in P_k$, define $f_{k+1}(e)=f_k(e)-\pi_k(P_k)$. For all other $e$, define $f_{k+1}(e)=f_k(e)$.
        \item If $f_{k+1}(e)=0$ for all $e \in \mathcal{E}$, terminate, else return to step 2.
    \end{enumerate}

    When this algorithm terminates, it produces a swap-progressive path decomposition. To see that it yields a path decomposition, note that at each iteration (i.e. for each $k$), we subtract $\min_{e \in P_k} f_k(e)$ from each $f_k(e)$ for $e \in P_k$. This means that firstly, $f_k \ge \mathbf{0}$ at each step of the algorithm. Second, at each iteration, we are zeroing the mass on some edge which previously had strictly positive weight. As the underlying graph is finite, we obtain that after finitely many steps, every edge will satisfy $f_k(e)=0$.

    We now argue this path decomposition is swap-progressive. Note that in each iteration of the algorithm, the path we construct always chooses the highest ranked available (i.e. with strictly positive weight) edge at each node it visits. Further, as discussed in the previous paragraph, the set of edges with strictly positive weights properly decreases with each iteration. It follows that, for $k < l$, whenever $P_k$ and $P_l$ visit the same node in the graph, $P_k$ leaves on a weakly higher edge than $P_l$ according to $\trianglerighteq$. This corresponds precisely to swap-progressivity.
\end{proof}

\subsubsection*{Proof of \autoref{thm:OrderedRUM}}

\begin{proof}
    In the random utility graph, each edge is of the form $(A, A\setminus\{x\})$. Assign edges of this form to the alternative $x$. Given the linear order over alternatives $\trianglerighteq_X$, choose any linear order $\trianglerighteq_\mathcal{E}$ over edges in the random utility graph such that $x \trianglerighteq_X y$ implies that $(A,A\setminus \{x\})\trianglerighteq_{\mathcal{E}}(B,B\setminus\{y\})$. Part (i) from \autoref{thm:OrderedRUM} now follows immediately from applying \autoref{thm:OrderedFlows} to the random utility graph using the linear order $\trianglerighteq_\mathcal{E}$.\footnote{Note that we could have used any order over edges for this construction. This translates to part (i) from \autoref{thm:OrderedRUM} holding even in the case of a menu dependent order over alternatives.}

    We now prove claim (ii). First, we show that a single-crossing representation is also a swap-progressive representation. Suppose that $\trianglerighteq_X$ is the linear order over alternatives for both the single-crossing and swap-progressive representations. Suppose towards a contradiction that $\succ_i$ and $\succ_j$, both in the support of the single-crossing representation, are non-trivially compatible, $i > j$, and that $x_j^{k+1} \trianglerighteq x_i^{k+1}$. Note that $x_j^{k+1} \succ_j x_i
    ^{k+1}$. This, the single-crossing property, $x_j^{k+1} \trianglerighteq x_i^{k+1}$, and $i > j$ tells us that $x_j^{k+1} \succ_i x_i^{k+1}$. However, we have that $x_i^{k+1} \succ_i x_j^{k+1}$, which is a contradiction. Thus any single-crossing representation is also swap-progressive.

    Now, it follows from \citet{filiz2023progressive} that a progressive representation is a random utility representation if and only if it is single-crossing. Thus, by the preceding paragraph, if a swap-progressive representation is single-crossing, it is also progressive. Since every swap-progressive representation is a random utility representation, it follows that every swap-progressive representation which is also progressive must be a random utility representation and thus single-crossing.
\end{proof}

\subsection{Parametric Invertibility Results}

\subsubsection{Topological Preliminaries}

Let $\mathcal{X}$, $\mathcal{Y}$ be topological spaces.  A continuous map $p: \mathcal{X} \to \mathcal{Y}$ is a {\bf covering map} if, for every $y \in \mathcal{Y}$ there exists a connected open neighborhood $V$ such that $p^{-1}(V)$ is a disjoint union of open sets $\{U_i\}_{i \in \mathcal{I}}$ in $\mathcal{X}$ each of which $p$ maps homeomorphically onto $V$. When $\mathcal{Y}$ is connected, the cardinality of the set $\mathcal{I}$ does not depend on the choice of $y\in \mathcal{Y}$, and $\vert \mathcal{I}\vert$ is referred to as the number of {\bf sheets} of the covering $p$. A homeomorphism $D: \mathcal{X} \to \mathcal{X}$ is called a {\bf deck transformation} if $(p \circ D)(x) = p(x)$ for all $x \in \mathcal{X}$. The set of all deck transformations for a covering map form a group under composition.\medskip

Fix $x_0 \in \mathcal{X}$ and $y_0 \in \mathcal{Y}$. The structure of a covering map is intimately tied to the fundamental groups $\pi_1(\mathcal{X}, x_0)$ and $\pi_1(\mathcal{Y}, y_0)$.\footnote{For the basic definitions and theory of fundamental groups, see, e.g. \cite{bredon2013topology} Chapter III.} We gather several relevant results below.

\begin{lemma}[\citealt{bredon2013topology}, Lemma III.4.4]\label{liftinglemma}
    Suppose $D$ is a deck transformation for $p$ such that $D(x) = x$ for some $x \in \mathcal{X}$. Then $D$ is the identity.\footnote{The statement here is consequence of Lemma III.4.4 in \cite{bredon2013topology} where, in his notation, $W = X$ and $f = p$. Alternatively, our statement here appears as an unlabeled remark at the top of p.\ 148.}
\end{lemma}

\begin{lemma}[\citealt{bredon2013topology}, Corr.\ III.5.3]\label{sheetlemma}
     If $\pi_1(\mathcal{X}, x_0)$ is trivial, then the cardinality of $\pi_1(\mathcal{Y}, y_0)$ equals the number of sheets of $p$.
\end{lemma}

\begin{lemma}[\citealt{bredon2013topology}, Corr.\ III.6.10]\label{decklemma}
    If $\pi_1(\mathcal{X}, x_0)$ is trivial, the group of deck transformations of $p$ is isomorphic to $\pi_1(\mathcal{Y}, y_0)$.
\end{lemma}

We say that a continuous map between topological spaces is {\bf proper} if the preimage of any compact set is compact. The following shows condition (ii) of \autoref{thm:paramthm} gives a simple characterization of properness.

\begin{lemma}[\citealt{lee2012smooth} Prop.\ 2.17]\label{properlemma}
    The map $\bar{\varphi}$ satisfies condition (ii) of \autoref{thm:paramthm} if and only if it is proper.
\end{lemma}

\noindent The final ingredient we require is a fixed-point theorem for periodic functions with prime period.\footnote{While \cite{smith1934theorem} does not explicitly state the requirement that $n$ be prime, this requirement is implicit in his proof. See also \cite{eilenberg1940theorem}, Footnote 2.}

\begin{theorem*}[\citealt{smith1934theorem}]\hypertarget{smiththeorem}{}
    Suppose $\mathcal{X} \subseteq \mathbb{R}^n$ is contractible, and $h: \mathcal{X} \to \mathcal{X}$ a homeomorphism such that the $n$-fold composition $(h \circ \cdots \circ h) = \textrm{id}$ for prime $n \in \mathbb{N}$. Then $h$ has a fixed point.
\end{theorem*}

\subsubsection{Proof of \autoref{thm:paramthm}}

\begin{proof}
    Necessity of (i) is standard, whereas (ii) follows from \autoref{properlemma}.  We now consider sufficiency. Condition (i) implies that $\bar{\varphi}$ is a local diffeomorphism (\citealt{lee2003smooth}, Prop.\ 4.8). In particular, this implies  $\bar{\varphi}$ is an open map, hence $\bar{\varphi}(\Theta) \subseteq \mathbb{R}^d$ is open and therefore a smooth manifold. Furthermore, \autoref{properlemma} implies $\bar{\varphi}$ is proper. Finally, as $\mathcal{X}$ is contractible, it is connected, hence so too is $\bar{\varphi}(\Theta)$. Then \cite{lee2003smooth} Prop.\ 4.46 implies $\bar{\varphi}: \Theta \to \bar{\varphi}(\Theta)$ is a smooth covering map.\medskip

    Fix some vector $y_0 \in \bar{\varphi}(\Theta)$. As $\bar{\varphi}$ is a covering map, the pre-image $\bar{\varphi}^{-1}(y_0)$ is discrete; since $\bar{\varphi}$ is proper, $\bar{\varphi}^{-1}(y_0)$ is finite, and by the connectedness of $\bar{\varphi}(\Theta)$, the preimage of every $y \in \bar{\varphi}(\Theta)$ has the same cardinality.  Now, as $\Theta$ contractible, it is simply connected, thus \autoref{sheetlemma} implies $\pi_1\big(\bar{\varphi}(\Theta), y_0\big)$ is finite, and by \autoref{decklemma}, $\pi_1\big(\bar{\varphi}(\Theta), y_0\big)$ acts on $\Theta$ via deck transformations such that the identity $\mathbbm{1} \in \pi_1\big(\bar{\varphi}(\Theta), y_0\big)$ is the only element mapped to the identity deck transform $\Theta \to \Theta$.\footnote{This action is simply the isomorphism whose existence is asserted by \autoref{decklemma}. Since the kernel of any group isomorphism is trivial, the claim follows.} Suppose now, for sake of contradiction, that $\bar{\varphi}(\Theta)$ is not simply connected. Then there exists some $\alpha \in \pi_1\big(\bar{\varphi}(\Theta), y_0\big)$, where $\alpha \neq \mathbbm{1}$. Since $\pi_1\big(\bar{\varphi}(\Theta), y_0\big)$ is finite, $\alpha$ must be of finite order, hence for some $n>1$, we have $\alpha^n = \mathbbm{1}$, and $\alpha^m \neq \mathbbm{1}$ for all $m < n$. Without loss of generality, we may suppose $n$ is prime.\footnote{If $n$ is not prime, let $m$ be any prime factor of $n$ and consider $\alpha^\frac{n}{m}$.} Then under the action of $\pi_1\big(\bar{\varphi}(\Theta), y_0\big)$ on $\Theta$, $\alpha$ corresponds to some non-identity deck transformation $D_\alpha: \Theta \to \Theta$ of prime period $n$. But by \hyperlink{smiththeorem}{Smith's Theorem}, $D_\alpha$ must have a fixed point, thus \autoref{liftinglemma} implies $D_\alpha$ is the identity, a contradiction. Thus $\pi_1\big(\bar{\varphi}(\Theta), y_0\big) = \{\mathbbm{1}\}$ and \autoref{sheetlemma} implies the preimage $\bar{\varphi}^{-1}(y_0)$, and hence $\bar{\varphi}^{-1}(y)$ for every $y \in \bar{\varphi}(\Theta)$, is singleton; in particular this means $\bar{\varphi}$ is injective.  By \cite{lee2003smooth} Prop.\ 4.33.b, $\bar{\varphi}$ is then a diffeomorphism, as desired.
\end{proof}

\pagebreak
\setcounter{page}{1}
\setcounter{footnote}{0}
\emptythanks
\title{Online Appendix for ``Identification in Stochastic Choice"}
\author{Peter Caradonna\thanks{Division of the Humanities and Social Sciences, Caltech.  Email: \href{mailto:ppc@caltech.edu}{\nolinkurl{ppc@caltech.edu}}.} \hspace{.1cm} and Christopher Turansick\thanks{Department of Decision Sciences, Bocconi University. Email: \href{mailto:christopher.turansick@unibocconi.it}{\nolinkurl{christopher.turansick@unibocconi.it}}.
}}
\date{\today}

\maketitle

\section{Beyond Random Utility}

\subsection{Random Choice Models}\label{sec:RCMgraph}\hypertarget{bru}{}

The proof of \autoref{thm:SCgeneral}, as well as the corresponding analogues of \autoref{thm:extremepoints} and \autoref{thm:OrderedRUM}, follow from \autoref{thm:flowryserspacetheorem}, \autoref{thm:Graphextremepoints}, and \autoref{thm:OrderedFlows} applied to the following graphical representation of random choice rules.\medskip

Let $\Sigma \subseteq 2^X \setminus \{\varnothing\}$ denote the collection of observable menus, and index these $A_1, \ldots, A_{\vert \Sigma \vert}$.  Define a graph with $\mathcal{N}=\Sigma \cup\{\varnothing\}$, and for every $1 \le i \le \vert \Sigma \vert$ and for each $x \in A_i$ there is a unique edge $A_i \to A_{i+1}$ in $\mathcal{E}$, where we define $A_{\vert \Sigma \vert +1}=\varnothing$.\footnote{Thus $\mathcal{E}$ is a multiset.} For the edge $e \in \mathcal{E}$ of the form $A_i \to A_{i+1}$ associated with the alternative $x \in A_i$, we define $f(e) = \rho(x,A_i)$. Since $\rho \geq0$ and $\sum_{x \in A_i}\rho(x,A_i)=1$, by \eqref{flowdef}, $f$ is a flow.\medskip

In the random choice graph, each path corresponds to a sequence of the form $(x_1,A_1),$ $(x_2,A_2),$ $\dots,$ $(x_{|\Sigma|},A_{|\Sigma|})$. This path exactly corresponds to the choice function which chooses $x_i$ from $A_i$. Thus there is a one-to-one correspondence between $\mathcal{P}$ and $\mathcal{C}_\Sigma$, and hence between path decompositions of $f$ and random choice representations. Moreover, a function $f: \mathcal{E} \to \mathbb{R}_+$ admits a path decomposition if and only if it is a (quasi-)flow, hence path decompositions on this graph precisely correspond to the identified sets of the random choice model.\medskip

We now proceed with the proof of \autoref{thm:SCgeneral}.
\begin{proof}  
    As an application of \autoref{MainLemmaflow} to the graphical representation described above, we have observational equivalence between two distributions over choice functions, $\nu$ and $\mu$, if and only if $\nu=\mu + \sum_{i=1}^n c_i R'_i$ where $R_i$ do not correspond to the Ryser swaps in \autoref{thm:SCgeneral} but the Ryser swaps from \autoref{MainLemmaflow}. For the purposes of this proof, call the Ryser swaps from \autoref{thm:SCgeneral} pseudo-Ryser swaps. We now show that every pseudo-Ryser swap can be generated by a series of Ryser swaps and every Ryser swap can be generated by a series of pseudo-Ryser swaps.

    Consider an unweighted Ryser swap. Given two choice functions $(c_0,c_1)$ it creates choice functions $(c_2,c_3)$ where $c_2$ agrees with $c_0$ strictly above $A_i$ (according to the indexing of sets used in our graphical representation) and agrees with $c_1$ weakly below $A_i$. Further, $c_3$ agrees with $c_1$ strictly above $A$ and agrees with $c_0$ weakly below $A_i$. Now our aim is to construct a pseudo-Ryser swap out of Ryser swaps. We consider a Ryser swap between $c_0$ and $c_1$ at choice set $A_i$. Start with two choice functions $(c_0,c_1)$. Ryser swap these choice functions at choice set $A_i$. This induces choice functions $(c_2,c_3)$. Consider the choice set $A_{i+1}$. Now Ryser swap $(c_2,c_3)$ at choice set $A_{i+1}$. The resulting choice functions $(c_4,c_5)$ satisfy the following.
    \begin{enumerate}
        \item $c_4$ agrees with $c_0$ on every set $B \neq A_i$ and agrees with $c_1$  at $A_i$
        \item $c_5$ agrees with $c_1$ on every set $B \neq A_i$ and agrees with $c_0$ at $A_i$
    \end{enumerate}
    Thus $(c_4,c_5)$ are the resultant choice functions of a pseudo-Ryser swap on $(c_0,c_1)$.

    Now we show that a Ryser swap can be constructed from pseudo-Ryser swaps. A Ryser swap swaps the choices between $c_0$ and $c_1$ at every choice set $A_j$ satisfying $j \geq i$ for some choice set $A_i$. Simply apply a pseudo-Ryser swap at each of these choice sets. This sequence of pseudo-Ryser swaps results in a Ryser swap. The prior equivalence shows that the span of Ryser swaps and the span of pseudo-Ryser swaps is the same. As such, by \autoref{MainLemmaflow}, we get that \autoref{thm:SCgeneral} holds.
\end{proof}

\subsection{Frame Dependent Random Utility}\label{sec:FRUMgraph}\hypertarget{frumgraph}{}
As mentioned in \autoref{sec:Extensions}, there is a more general version of the frame-dependent random utility model. Instead of an alternative being framed and not framed, there are now multiple frames. We label these frames as $\{1,\dots,M\}$. In addition, there is a second version of the model where the frames are labeled as $\{0,1,\dots,M\}$ and the zero frame corresponds to an alternative not being available. For now, we will focus on the version of the model without the zero frame.

Let $x^m$ denote alternative $x$ under frame $m$. Mirroring our notation from \autoref{sec:Extensions}, define $X^*=\{x_1^1,\dots,x_K^1,\dots,x_1^M,\dots,x_K^M\}$ denoting the set of alternative-frame pairs. A menu corresponds to some subset $A \subsetneq X^*$ satisfying, for each $i\in \{1,\dots,K\}$, there is exactly one $m$ such that $x_i^m\in A$. Thus a menu contains every alternative $x_i$ with each alternative under some alternative-dependent frame $m$.

The frame-dependent random utility model considers strict preferences over the set $X^*$ under one additional assumption. The model restricts to preferences $\succ$ such that, if $m>m'$, then $x_i^m \succ x_i^{m'}$. Given a preference $\succ$, observe that if, for $m>1$ and arbitrary $i\neq j$, $x_j^1 \succ x_i^m$, then there is no menu $A$ in which $x_i^m$ is maximal. With this in mind, we now define a truncated preference $\succ^*$. Given a preference $\succ$, let $\succ^*$ be the ranking which agrees with $\succ$ up to, and including, the best alternative under frame $1$. That is, $\succ^*$ is formed by taking $\succ$ and removing all alternatives $x_i^m$ which rank below the highest alternative of the form $x_j^1$. Let $\mathcal{L}^*$ denote the set of truncated strict preferences.

The data are compatible with the frame-dependent random utility model if there exists some distribution $\mu \in \Delta(\mathcal{L}^*)$ such that: $$\rho(x,A)=\sum_{\succ^*\in\mathcal{L}^*}\mu(\mathcal{L}^*)\mathbbm{1}_{\{x \textrm{ is $\succ^*$-maximal in $A$}\}}$$ for all menus $A$ and each $x \in A$. Here, $x$ corresponds to some alternative $x_i^m$ and menu $A$ satisfies the assumption we have made in this section. As before, we say that two truncated preferences $\succ^*_1,\succ_2^*\in\mathcal{L}^*$ are $k$-compatible if they agree on their $k$-best alternatives, though not necessarily their ranking. A Ryser swap for the frame-dependent model is simply a signed measure in $\mathbb{R}^{\mathcal{L}^*}$ which places unit mass on a $k$-compatible pair $\succ_1^*$ and $\succ_2^*$ and mass negative one on the pair of truncated preferences obtained by swapping their $k$-initial segments. Let $\mathcal{R}^{fd} \subset \mathbb{R}^{\mathcal{L}^*}$ denote the linear span of these vectors. \autoref{thm:FRUMcor} still holds under this formulation of the frame-dependent random utility model and $\mathcal{R}^{fd}$.

The graphical representation of the frame dependent random utility model is due to \citet{Cheung2024FRUM} although we make a slight alteration to fit exactly into our framework. Let $\mathcal{X}$ denote the collection of all menus $A \subsetneq X^*$ satisfying, for each $i\in \{1,\dots,K\}$, there is exactly one $m$ such that $x_i^m\in A$. Consider the directed acyclic graph where $\mathcal{N}=\mathcal{X} \cup\{t\}$ and the edge set is described as follows.
\begin{itemize}
    \item For each $A \in \mathcal{X}$ and each $x_i^m \in A$ with $i>1$, there is an edge from $A$ to $(A \setminus \{x_i^m\})\cup\{x_i^{m-1}\}$, $A\rightarrow(A \setminus \{x_i^m\})\cup\{x_i^{m-1}\}$
    \item For each $A \in \mathcal{X}$ and each $x_i^1 \in A$, there is an edge from $A$ to $t$, $A\rightarrow t$
\end{itemize}
Here, the source node corresponds to the node $A^*$ where $A^*$ is the menu given by $\{x_1^M,\dots,x_K^M\}$. In order to describe our flow function, first observe that each menu $A$ can be represented as a vector in $\bar{A} \in \mathbb{R}^X$. Each dimension of $\bar{A}$ is indexed by some element $x \in X$. The dimension of $\bar{A}$ corresponding to alternative $x$ is given by the frame associated with $x$ in menu $A$. Thus, if $x_i^m\in A$ then the $i$-th dimension of $\bar{A}$ is $m$. Under this observation, the vectors $\bar{A}$ can be partially ordered by the greater than or equal to relationship. As an abuse of notation, we rank menus $A \geq B$ if $\bar{A}\geq \bar{B}$. Given the frame-dependent random choice rule $\rho(x,A)$, we can implicitly define a new object:
$$\rho(x,A)=\sum_{A \leq B: \bar{A}_x=\bar{B}_x}y(x,B),$$
where $\bar{A}_x$ denotes the entry in vector $\bar{A}$ in the dimension indexed by $x$. \citet{Cheung2024FRUM} provides a closed form expression for $y(x,A)$, but, to keep things simple, we only present the recursive definition here. For our graphical representation, we use $y(x,A)$ as the flow assigned to the edge leaving node $A$ which corresponds to alternative $x$. \citet{Cheung2024FRUM} shows that, whenever the frame-dependent random choice rule has a frame-dependent random utility representation, this constitutes a flow as in \autoref{flowdef}.

Our interpretation of the frame-dependent random utility graph relies on our notion of a truncated preference. Consider the path given by $A^*\rightarrow B \rightarrow \dots \rightarrow t$. Observe that $B$ differs from $A^*$ by a single element and only by the frame of that element. Specifically, $A^*$ has $x_1^M$ while $B$ has $x_1^{M-1}$. This generalizes to successive nodes in our path. We now rewrite our path with superscripts on each $\rightarrow$ corresponding to the difference between these two successive nodes: $$A^*\rightarrow^{x_1^M} B \rightarrow^{x_i^m} \dots \rightarrow^{x_j^1} t.$$ This path is then associated with the truncated preference $x_1^M \succ^* x_i^m\succ^*\dots \succ^*x_j^1$. Just as in the previous cases, a path decomposition putting mass on a path corresponds to a probability distribution putting mass on its corresponding truncated preference. \autoref{thm:FRUMcor} is then recovered by applying \autoref{thm:flowryserspacetheorem} to the frame-dependent random utility graph.

As we mentioned earlier, there is a second version of the frame-dependent random utility model which allows for variation in which alternatives are available. This can be encoded as a zero frame which places sufficiently negative utility on an alternative in order to make it never be chosen. The graphical representation which \citet{Cheung2024FRUM} provides for this version of the model satisfies the assumptions we make on the graphs we consider in \autoref{thm:flowryserspacetheorem}. As such, our results can be applied directly out of the box to the frame-dependent random utility model with availability variation. In this case, types correspond to full, strict preferences (i.e. not truncated) over $X^*$.

\subsection{Dynamic Discrete Choice}\label{sec:DDCgraph}
In this section, we extend the setup we described in \autoref{sec:DDC} to allow each period's menu to depend on the previous period's choice. There is a sequence $t = 1, \ldots, T$ of discrete time periods. In each, an agent chooses from some menu $\varnothing \subsetneq A_t \subseteq X$. We assume $A_1$ is exogenously given and, for all $2 \le t \le T-1$, that the menu $A_t$ faced by the agent at time $t$ is determined by:
\[
    A_t(x_{t-1}) \equiv g_{t-1}(x_{t-1}),
\]
where $x_{t-1} \in X$ denotes the agent's choice at $t-1$, and $g_{t-1}: X \to 2^X \setminus \{\varnothing\}$ is a fixed evolution function, describing how the agent's choice at $t-1$ determines their choice set at time $t$.

We consider an empiricist who observes a system of conditional choice probabilities. Formally, this consists of (i) a distribution of time-one choice probabilities $\rho_1 \in \Delta(X)$ supported on $A_1$, and (ii) a family of functions $\rho_t: X \times X \to [0,1]$, for $2 \le t \le T$, such that:
\[
        \sum_{a \in A_t(b)} \rho_t(a \vert b) = 1,
\]
whenever $\rho_{t-1}(b \vert \, \cdot \,)$ is not everywhere zero, and $\rho_t(a \vert b) = 0$ for all $x \in X$ otherwise. Given such data, we seek to describe the set of compatible distributions over choice \emph{histories}.

To represent this problem as one of describing flow functions, we define  $\mathcal{N}$ recursively, as the set of all pairs in $(t,x)$ such that either (i) $t=1$ and $x \in A_1$, or (ii) $2 \le t \le T$, and $x \in A_t(x')$ for some $(t-1, x') \in \mathcal{N}$, as well as an abstract source and sink, $s$ and $t$. Similarly, we let $\mathcal{E}$ consist of all directed edges of the form $(t-1, x') \to (t,x)$, where $x \in A_t(x')$, as well as edges from the source $s$ to each pair in $\{1\} \times A_1$, and edges from each pair in $\mathcal{N}$ with $t=T$ to the sink $t$.

Given such a graph, we recursively define the flow from the data $\rho_t$. For the edge connecting the source node to $(1,x)$, we assign $\rho_1(x)$ as the flow assigned to that edge. For edges connecting $(t-1,x')$ to $(t,x)$, we normalize $\rho_{t}(x|x')$ by multiplying by the total in-flow into node $(t-1,x')$ and assign this new value as the flow for this edge. Finally, for the edge connecting $(T,x)$ to the sink, we simply assign the total in-flow into node $(T,x)$ as the flow. It is straightforward to verify that our construction defines a flow. In fact, our construction recovers the absolute frequency of conditional choices in each period. Each path in our constructed graph corresponds to a full history of choices and a flow function of this graph assigns frequencies to each full history.

\autoref{thm:DDCObsEq} can be recovered by applying \autoref{MainLemmaflow} to this graphical representation under the assumption that the choice set is fixed across time periods and $g_{t-1}(\cdot)$ is equal to our fixed choice set. A version of \autoref{thm:DDCObsEq} can be recovered for the more general setup described here. We simply need to restrict ourselves to considering swaps between histories which are possible given $g_{t-1}(\cdot)$.

\subsection{Correlated Random Utility}
In addition to the three types of models we propose in \autoref{sec:Extensions}, our techniques can also be used to study identification in the correlated random utility model of \citet{chambers2024correlated} as well as the consumption dependent random utility model of \citet{turansick2024consumption}. Suppose that an analyst has access to a random joint choice rule, $\rho: X^2 \times (2^X\setminus\{\varnothing\})^2\rightarrow [0,1]$, which satisfies, for all non-empty $A,B \subseteq X$, $$\sum_{x\in A}\sum_{y \in B} \rho(x,y,A,B)=1.$$ The \textbf{correlated random utility model} asks that there is distribution $\mu \in \Delta(\mathcal{L}^2)$ such that $$\rho(x,y,A,B)=\sum_{\succ \in \mathcal{L}} \sum_{\succ' \in \mathcal{L}} \mu(\succ,\succ') \mathbbm{1}_{\{x \; \succ \,A \setminus x,\; y \; \succ' B\setminus y\}}.$$
The correlated random utility model represents choice in two periods where choices in the two periods may be correlated through the distribution over preference pairs, but, conditional on this draw, the agent's choice in the first period has no impact on their choice in the second period. On the other hand, the \textbf{consumption dependent random utility model} asks that there is a distribution $\mu \in \Delta(\mathcal{L})$ as well as distributions for each alternative preference pair, $(x,\succ)$, $\nu_{(x,\succ)}\in \Delta(\mathcal{L})$ such that $$\rho(x,y,A,B)=\sum_{\succ \in \mathcal{L}}\mu(\succ) \mathbbm{1}_{\{x \; \succ \,A \setminus x\}}\sum_{\succ' \in \mathcal{L}} \nu_{(x,\succ)}(\succ')\mathbbm{1}_{\{y \; \succ' \;B \setminus y\}}.$$
The consumption dependent random utility model represents choice in two periods where an agent's preference, and thus choice, in the second period depends both on their preference and choice in the first period. In each of these models, first period choices are well defined independently of the second period's choice set. As such, we use $\rho_1(x,A)=\sum_{y\in B}\rho(x,y,A,B)$ for any choice of non-empty $B$. 

Both of these models share the same graphical representation, but, while every path on this graph corresponds to some type in the consumption dependent random utility model, the correlated random utility model constitutes a support restriction on this graph. The graphical representation for these two models consists of a series of graphs, one representing choice in the first period and the rest, one for each edge of the first period graph, representing choice in the second period. These series of graphs can then be made into a single graph by replacing each edge of the first period graph with the corresponding second period graph.

To begin, the first period graph is exactly the random utility graph with flows coming from $\rho_1$ and \autoref{bm}. Each second period graph has the same node and edge set as the random utility graph. That is $\mathcal{N} =  2^{X}$ and $\mathcal{E}=\big \{(A,B)\ \vert \; B=A\setminus \{a\}, a \in A\big \}$, where $s = X$ and $t = \varnothing$. The second period graphs differ from the first period graphs in their quasi-flow function. For the second period graph associated with the edge $(A \rightarrow A \setminus \{a\})$ in the first period graph, the random joint choice rule $\rho$ induces a unique flow as follows:
\begin{equation}\label{bmcrum}
    f(B \rightarrow B \setminus \{b\})=\sum_{A \subseteq A'} \sum_{B \subseteq B'}(-1)^{|A'\setminus A| + |B'\setminus B|} \rho(a,b,A',B').
\end{equation}
The function $f$ given in \autoref{bmcrum} defines a quasi-flow as shown in \citet{chambers2024correlated}.

In the consumption dependent random utility model, a type corresponds to a preference in the first period and $|X|$ preferences in the second period. Each of these $|X|$ preferences correspond to the type's preference in the second period conditional on choosing alternative $x$ in the first period. As is the case with the random utility graph, in the first period graph, as well as each of the second period graphs, each path is bijectively associated with a preference. Recall that this bijection is given by regarding a path:
\[
    X \to X \setminus \{x_1\} \to \cdots \to X \setminus \{x_1, \ldots, x_{N-1}\} \to \varnothing
\]
as a sequence of nested lower contour sets, which uniquely defines the preference $x_1 \succ \cdots \succ x_N$. With this bijection in mind, the path on the first period graph defines the preference which dictates the agent's first period choice. Given this path and edge $(A,A\setminus\{a\})$ in this path, a path on the second period graph associated with edge $(A,A\setminus\{a\})$ on the first period graph corresponds to the agent's second period preference conditional on choosing $a$ in the first period. The consumption dependent random utility model puts no restriction on which paths are allowed in the second period graph while the correlated random utility model asks that, for a given type, all paths on the second period graphs are the same.

This graphical construction is most easily understood via a first period graph and a series of second period graphs. We can directly apply our result about path decomposition to each of these graphs to discuss uniqueness in the correlated and consumption dependent random utility models. As mentioned earlier, we can alternatively replace each edge in the first period graph with its corresponding second period graph to get a single graph which acts as a graphical representation for these two models. In this case, we need only apply our graphical results to this single graph.

While our focus in this section was to discuss how to extend the graphical representation of the random utility model in order to capture multidimensional analogues of the random utility model, this process generalizes to other models which can be represented by path decompositions of quasi-flows on directed acyclic graphs. This allows for discussion of both correlation of types/paths across dimension, as in the correlated random utility model, as well as dependence of later types/paths on earlier types-choice/path-edge pairs, as in the consumption dependent random utility model.

\subsection{Maximization of Other Types of Orders}
In the main body of the paper, we focused on agents who maximize linear orders. However, an analyst may be wary of this assumption and instead only ask that agents maximize a weak order, an interval order, or a semiorder. Under differing data assumptions than what we use for the random utility model, \citet{davis2018extended} provides graphical representation for models which consider a distribution over agents maximizing each of these three types of orders. \citet{doignon2023adjacencies} also studies the graphical representation of these models in order to characterize adjacency of extreme points in their corresponding polytopes. The techniques we develop here can be applied to discuss uniqueness in each of these models.

\section{Further Results}

\subsection{Technical Results}

\begin{lemma}\label{rationalLem}
    Let $\mathcal{R}$ denote the subspace of $\mathbb{R}^\mathcal{P}$ spanned by vectors of the form:
    \[
        \mathbbm{1}_{\{P, P'\}} - \mathbbm{1}_{\{P'', P'''\}}
    \]
    where $(P,P')$ are compatible and $(P'', P''')$ are their corresponding conjugate, and by minor abuse of notation let $\Delta(\mathcal{P}')$ denote the face of $\Delta(\mathcal{P})\subseteq \mathbb{R}^\mathcal{P}$ spanned by the abstract simplex $\varnothing \subsetneq \mathcal{P}' \subseteq \mathcal{P}$.  Suppose there exists $t \in \mathbb{R}^\mathcal{P}$ such that:
    \[
       \dim \bigg[\big(\mathcal{R} + t\big) \cap \Delta(\mathcal{P}')\bigg] > 0.
    \]
    Then there exists $t' \in \mathbb{R}^\mathcal{P}$ such that $\big(\mathcal{R} + t'\big) \cap \Delta(\mathcal{P}')$ contains a pair of distinct points in $\mathbb{Q}^\mathcal{P}$.
\end{lemma}
\begin{proof}
    Let $\mathcal{V}$ denote the subspace of $\mathbb{R}^\mathcal{P}$ given by:
    \[
        \textrm{span } \big\{e^{P} - e^{P'}\big\}_{P,P' \in \mathcal{P}'},
    \]
    where $e^i$ denotes the $i$th standard Euclidean  basis vector.  By hypothesis, 
    \[
        K = \textrm{dim } \mathcal{R} \cap \mathcal{V} > 0.
    \]
    From their definitions, both $\mathcal{R}$ and $\mathcal{V}$ admit bases $\{q_{\mathcal{R}}^i\}_{i = 1}^{\textrm{dim}(\mathcal{R})}$ and $\{q_{\mathcal{V}}^j\}_{j = 1}^{\textrm{dim}(\mathcal{V})}$ which belong to $\mathbb{Q}^\mathcal{P}$. Define the $\vert \mathcal{P}\vert  \times \textrm{dim}(\mathcal{R})$ and $\vert \mathcal{P}\vert \times \textrm{dim}(\mathcal{V})$ matrices:
    \[
        Q_\mathcal{R} = \begin{bmatrix} q^1_\mathcal{R} & \cdots & q^{\textrm{dim}(\mathcal{R})}_\mathcal{R} \end{bmatrix}
    \]
    and
    \[
        Q_\mathcal{V} = \begin{bmatrix} q^1_\mathcal{V} & \cdots & q^{\textrm{dim}(\mathcal{V})}_\mathcal{V} \end{bmatrix},
    \]
    respectively, and let:
    \[
        Q = \begin{bmatrix} Q_\mathcal{R} & \aug & -Q_\mathcal{V} \end{bmatrix}
    \]
    Since each of the above matrices consists exclusively of rational elements, their respective column spaces trivially admit a bases in $\mathbb{Q}^\mathcal{P}$. This implies that, by Gaussian elimination, the annihilators of their column spaces admit a bases in $\mathbb{Q}^\mathcal{P}$; in particular the annihilator of the column space of $Q$ admits a rational basis $\{r^K\}_{k=1}^K$, where each vector $r^k = [r^k_\mathcal{R} \; \vert \; r^k_\mathcal{V}]$.\medskip

    Let $\bar{q}^k = Q_\mathcal{R} r^k_\mathcal{R} \, \big( = Q_\mathcal{V} r_\mathcal{V}^k\big)$.  By construction, $\{\bar{q}^k\}_k \subset \mathbb{Q}^\mathcal{P}$ and, by construction $\{\bar{q}^k\}_k$ form a basis for $\mathcal{R} \cap \mathcal{V}$. Since $\textrm{dim}(\mathcal{R} \cap \mathcal{V}) > 0$, there exists some non-zero $\bar{q}^k \in \mathcal{R} \cap \mathcal{V} \cap \mathbb{Q}^\mathcal{P}$. Then letting:
    \[
        t' = \frac{1}{\vert \mathcal{P}'\vert} \mathbbm{1}_{\mathcal{P}'},
    \]
    we have $t'$ and $t' + \alpha \bar{q}^k$, for small enough $\alpha \in \mathbb{Q}_{++}$, are distinct rational vectors in $\big(\mathcal{R} + t'\big) \cap \Delta(\mathcal{P}')$ as desired.
\end{proof}

\subsection{Order-Based Restrictions and Extreme Points}\label{sec:orderedExtpoints}\hypertarget{sec:orderedExtpoints}{}
We begin by noting that every swap-progressive path decomposition of a quasi-flow is an extreme point of the set of path decompositions for the given quasi-flow.
\begin{proposition}\label{prop:swapProgisExt}
    Consider a graph with quasi-flow function $f$. Each swap-progressive path decomposition of $f$ is an extreme point of the set of path decompositions of $f$.
\end{proposition}
\begin{proof}
    Suppose we have a swap-progressive path decomposition $\pi$ for some quasi-flow function $f$. Note that any path decomposition with the same support is also a swap-progressive path decomposition as swap-progressivity is a property of the support. By \autoref{thm:OrderedFlows}, we know that a swap-progressive path decomposition is unique. This means that no other path decomposition whose support is contained within $\textrm{supp}(\pi)$ decomposes the same quasi-flow. Thus, the set of flows $\{\mathbbm{1}_{\{e\in P\}}\}_{P\in \textrm{supp}(\pi)}$ is linearly independent. By \autoref{cor:LinInd} and \autoref{thm:Graphextremepoints}, we get that $\pi$ is an extreme point of the set of path decompositions inducing quasi-flow $f$.
\end{proof}

Now we note that our notion of swap-progressivity, for both graphs and random utility, is a special case of a more general notion of an ordered restriction. Consider a graph $(\mathcal{N},\mathcal{E},s,t)$ and the set of paths $\mathcal{P}$. Consider some linear order over $\mathcal{P}$, $\trianglerighteq_\mathcal{P}$. For ease, $\trianglerighteq_\mathcal{P}$ is equivalent to an enumeration of $\mathcal{P}$ where $P_i \trianglerighteq_\mathcal{P} P_j \iff i \leq j$. Given an ordering over edges $\trianglerighteq_\mathcal{E}$, we can induce some ordering over paths $\trianglerighteq_\mathcal{P}$. We do this by starting at node $s$ and saying that a path $P \trianglerighteq_\mathcal{P}P'$ if $(s,n) \trianglerighteq_\mathcal{E}(s,n')$ where $(s,n)$ and $(s,n')$ are the first edges of $P$ and $P'$ respectively. Generally, this will not yet leave us with a full linear over paths so we iterate on this construction. Now move onto any node $n$ for which each edge of the form $(m,n)$ has already been considered. At node $n$, each path $P$ which passes through node $n$ will already be partially ordered. Further, by construction, incomparability by this partial order is an equivalence relation among those paths which pass through node $n$. For each of these equivalence classes, say that $P \trianglerighteq_\mathcal{P}P'$ if $(n,m) \trianglerighteq_\mathcal{E}(n,m')$ where $(n,m)$ and $(n,m')$ are edges in $P$ and $P'$ respectively. By iterating this process through each node, we are left with a linear order over paths, $\trianglerighteq_\mathcal{P}$. Recall that our construction of a swap-progressive proceeds as follows.
\begin{enumerate}
    \item Set $i=1$.
    \item Put as much mass on $P_i$ as possible (i.e. the minimum of $f(e)$ among edges in $P_1$). Call this amount $\pi(P_i)$.
    \item Set $f(e)=f(e)$ if $e \not \in P_i$. Set $f(e)=f(e)-\pi(P_i)$ if $e \in P_i$.
    \item Set $i=i+1$. If $i>|\mathcal{P}|$, terminate the construction. Otherwise, return to step $2$.
\end{enumerate}
Notably, the construction works for any linear order over paths $\trianglerighteq_\mathcal{P}$ and, given a quasi-flow function $f$, the induced path decomposition of the quasi-flow function will be unique given $\trianglerighteq_\mathcal{P}$. We call a path decomposition induced by such a construction an \textbf{ordered path decomposition}. Single-crossing, swap-progressive, and progressive representations are ordered path decompositions of the correctly chosen graph, but there are also ordered path decompositions which are not single-crossing, swap-progressive, or progressive. While ordered path decompositions share many of the nice identification properties of swap-progressivity, ordered path decompositions are not sufficiently general enough to recover every extreme point of the set of path decompositions for every quasi-flow.

\begin{proposition}\label{prop:ExtnotOrdered}
    There exist graphs $(\mathcal{N},\mathcal{E},s,t)$, flows $f$, and extreme points of the set of path decompositions of $f$ such that any ordered decomposition of $f$ does not induce the extreme point. Further, the random utility graph has this property.
\end{proposition}

\begin{proof}
    Consider the graph in Figure 2 of \citet{chambers2024limits} and the uniform distribution over the eight paths considered in Example 4.1 of the same paper, i.e. putting $\frac{1}{8}$ mass on each of the eight paths. Chambers \& Turansick show the indicators of these eight paths form a linearly independent set of flows. As such, by \citet{winkler1988extreme}, the uniform distribution is an extreme point of the set of edge decompositions inducing the corresponding flow. Now, fix any ordered path decomposition $\pi$. By definition (i.e. the construction in the paragraph preceding \autoref{prop:ExtnotOrdered}), the highest ranked path in $\textrm{supp}(\pi)$ must contain some edge unique to it among paths in $\textrm{supp}(\pi)$. However, as every path in Example 4.1 of \citet{chambers2024limits} shares each of its edges with at least one other path in the example, the uniform distribution over these paths cannot arise as an ordered path decomposition, for any choice of order $\trianglerighteq_\mathcal{E}$; as the graph in this example is obtained as a support restriction of the random utility graph, we are done.
\end{proof}

\section{Examples and Computations Omitted From Text}

\subsection{Failure of Identification in \autoref{chriscex}}\hypertarget{chriscexworked}{}

Let $X=\{x_0, x_1, x_2\}$ and consider any random choice rule $\rho$ which satisfies:
\[
    \textrm{(i)} \quad \rho(x_1,x_0x_1)=\rho(x_2,x_0x_2)=\frac{1}{11} \quad \textrm{ and } \quad \textrm{(ii)} \quad \rho(x_1,x_0x_1x_2)=\rho(x_2,x_0x_1x_2).
\]
(Recall here we only consider menus which include the outside option $x_0 \in X$, hence this completely specifies $\rho$ up to the choice frequency of $x_0$ from $X$.) Consider again the submodel $s$ defined in \autoref{chriscex}. For $i = 1,2$ we have:
\[
    s^i(v_1,c_1, v_2,c_2) = \frac{1+ v_i}{v_i (1+v_ic_i)} = 10,
\]
or, simplifying:
\[
    v_i^2 c_i = 10
\]
for all $i = 1,2$. Similarly, by hypothesis we have:
\begin{equation}\label{keqn}
    s^{i+2}(v_1,c_1,v_2,c_2) =  \frac{1+v_i+v_{i+1}}{v_i+v_iv_{i+1}+ c_iv_i^2} = K
\end{equation}
for each $i = 1,2$ and some $K > 0$, hence simplifying yields:
\[
    v_1 = v_2.
\]
Thus, there are only two pieces of information we have not determined yet: the value of $v_1$ (without loss), and the value of $K$ in \eqref{keqn}, which pins down the choice probability of the outside option from the menu $X$.  From \eqref{keqn}, we obtain:
\[
    \frac{1+2v_1}{v_1 + v_1^2 + 10} = K,
\]
or:
\[
    K\big(10+v_1 + v_1^2) = 1+2v_1.
\]
When, e.g., $K = 0.3$ this has multiple solutions, i.e.\ $v_1 = \frac{5}{3}$ and $v_1 = 4$.\hfill $\blacksquare$

\subsection{A Non-Single Crossing Example}\hypertarget{OrderedFish}{}

Let $X=\{a,b,c,d\}$, and consider the random choice rule $\rho$, which chooses uniformly from any menu. For any linear ordering $\trianglerighteq$ of alternatives, $\rho$ fails to satisfy the centrality axiom of \cite{apesteguia2017single}, and hence does not admit a single-crossing representation.  However, $\rho$ is clearly consistent with the random utility model: for example it is rationalized by the uniform distribution on $\mathcal{L}$.\medskip

By \autoref{thm:OrderedRUM}, the random choice rule $\rho$ admits a swap-progressive rationalization for any order $\trianglerighteq$. To see how swap-progressivity pins down a unique representation, first note that the uniform distribution on $\mathcal{L}$ places positive (and equal) mass on six distinct, non-overlapping collections of four preferences, each consisting of a 2-compatible pair and its 2-conjugate.\footnote{These collections each consist of, for any $\{x,y\} \subset X$, the four preferences whose initial strings are $xy$ or $yx$.}  Given $\trianglerighteq$, the unique swap-progressive representation places (equal) mass on only two of the four preferences in each collection, namely those satisfying:
\[
    \textrm{(i)} \quad \big\{x^1 \trianglerighteq x^2  \text{ and } x^3 \trianglerighteq x^4\big\} \quad \textrm{ or  } \quad \textrm{(ii)} \quad \big\{x^2 \trianglerighteq x^1 \text{ and } x^4 \trianglerighteq x^3\big\}.
\]
In other words, on each collection, the swap-progressive representation uses the order $\trianglerighteq$ to select one of the two compatible pairs to receive full mass.\footnote{See also \autoref{ex:OrderedFishConstruction}.} Concretely, if $a\trianglerighteq b \trianglerighteq c \trianglerighteq d$, the unique swap-progressive representation places equal mass on:
\[
    \succ_1: dcba \quad \succ_2: dbca \quad \succ_3: dabd \quad \succ_4: cdab \quad \succ_5: cbda \quad \succ_6: cadb
\]
\[
    \succ_7: bdac \quad \succ_8: bcad \quad \succ_9: badc \quad \succ_{10}: adbc \quad \succ_{11}: acbd \quad \succ_{12}: abcd.
\]
Thus, even when a single-crossing representation may fail to exist, swap-progressivity provides a natural criteria for selecting rationalizations in a manner that incorporates the natural ordering of the environment.\hfill $\blacksquare$

\end{document}